\documentclass{article}
\usepackage[utf8]{inputenc}
\usepackage[margin=.94in]{geometry}
\usepackage[dvipsnames]{xcolor}
\usepackage{graphicx,epstopdf}
\usepackage{amsmath,amssymb,amsthm,bm,bbm}
\usepackage{mathtools}
\usepackage{braket}
\usepackage{csquotes}
\usepackage{subcaption}
\usepackage{booktabs}
\usepackage{times}

\usepackage[colorlinks=true,linkcolor=blue,bookmarks=true,breaklinks=true, citecolor=blue, urlcolor=blue]{hyperref}
\usepackage[numbers,sort&compress]{natbib}
\makeatletter \AtBeginDocument{%
\immediate\write\@auxout{\string\citation{REVTEX42Control}}%  
\immediate\write\@auxout{\string\citation{apsrev42Control}}%
}
\makeatother

\usepackage[capitalize,nameinlink]{cleveref}
\AddToHook{cmd/appendix/before}{%
    \crefalias{section}{appendix}%
    \crefalias{subsection}{appendix}
}

\usepackage[normalem]{ulem}
\usepackage{makecell}
\usepackage{authblk}

\theoremstyle{plain}
\newtheorem{theorem}{Theorem}
\newtheorem{lemma}[theorem]{Lemma}

\newtheorem{proposition}[theorem]{Proposition}
\theoremstyle{definition}

\let\originalleft\left
\let\originalright\right
\renewcommand{\left}{\mathopen{}\mathclose\bgroup\originalleft}
\renewcommand{\right}{\aftergroup\egroup\originalright}

\usepackage{dsfont}

\newcommand{\e}{\textup{e}}
\renewcommand{\i}{\textup{i}}

\newcommand{\id}{\mathds{1}}

\newcommand{\R}{{\mathbb{R}}}

\newcommand{\ad}{\textup{ad}}

\newcommand{\abs}[1]{\left\lvert#1\right\rvert}
\newcommand{\norm}[1]{\left\lVert#1\right\rVert}

\newcommand{\order}[1]{\mathcal{O}(#1)}
\usepackage{multirow}

\title{Faster Algorithmic Quantum and Classical Simulations\\ by Corrected Product Formulas}
\author[1,2]{Mohsen Bagherimehrab\footnote{Corresponding author: mohsen.bagherimehrab@gmail.com}}
\author[2,4]{Luis Mantilla Calder\'on}
\author[3]{Dominic W. Berry}
\author[2,4]{Philipp Schleich}
\author[1,2]{Mohammad Ghazi Vakili}
\author[1]{Abdulrahman Aldossary}
\author[1,4]{Jorge A. Campos Gonzalez Angulo}
\author[5,6]{Christoph Gorgulla}
\author[1,2,4,7,8,9]{Al\'an Aspuru-Guzik}
\affil[1]{Chemical Physics Theory Group, Department of Chemistry, University of Toronto, Toronto, Ontario, Canada} 
\affil[2]{Department of Computer Science, University of Toronto, Toronto, Ontario, Canada}
\affil[3]{School of Mathematical and Physical Sciences, Macquarie University, Sydney, NSW 2109, Australia}
\affil[4]{Vector Institute for Artificial Intelligence, Toronto, Ontario, Canada}
\affil[5]{Department of Physics, Harvard University, Cambridge, MA, USA}
\affil[6]{Department of Structural Biology, St. Jude Children's Research Hospital, Memphis, TN, USA}
\affil[7]{Department of Chemical Engineering \& Applied Chemistry, University of Toronto, Toronto, Ontario, Canada}
\affil[8]{Department of Materials Science \& Engineering, University of Toronto, Toronto, Ontario, Canada}
\affil[9]{Lebovic Fellow, Canadian Institute for Advanced Research, Toronto, Ontario, Canada}
\begin{document}
\maketitle
\begin{abstract}
    Hamiltonian simulation using product formulas is arguably the most straightforward and practical approach for algorithmic simulation of a quantum system's dynamics on a quantum computer. Here we present corrected product formulas (CPFs), a variation of product formulas achieved by injecting auxiliary terms called correctors into standard product formulas. We establish several correctors that improve the accuracy of standard product formulas by orders of magnitude when simulating Hamiltonians comprised of two exactly simulatable partitions, a common structure of lattice Hamiltonians. Importantly, injecting these correctors increases the overall simulation cost by only a small additive or multiplicative factor. We show that correctors are particularly advantageous for perturbed systems, where one partition has a relatively small norm compared to the other, as they allow the small norm to be utilized as an additional parameter for controlling the simulation error. We demonstrate the performance of CPFs by numerical simulations for several lattice Hamiltonians. Numerical results show that our theoretical error bounds for CPFs match or outperform the empirical error scaling of standard product formulas for these systems. We also demonstrate improvements offered by CPFs by implementing small-size systems on actual quantum hardware, as well as on noisy and noiseless quantum simulators. CPFs could be a valuable algorithmic tool for early fault-tolerant quantum computers with limited computing resources. As for standard product formulas, CPFs could also be used for simulations on a classical computer.
\end{abstract}

\newpage
\tableofcontents

%==================================================
\section{Introduction}
%==================================================

Simulating the dynamics of quantum systems is arguably the most natural application of quantum computers. Indeed, quantum computation was originally motivated by the problem of quantum simulation~\cite{fey82,benioff_computer_1980,manin_computable_1980}. The solution to this problem will allow us to probe the foundational theories in physics, chemistry, and materials science, ultimately leading to potential practical applications such as designing new pharmaceuticals, catalysts, and materials~\cite{ADL+05,KJL+08,Chan24,NMV+24}.
The original proposal for simulating quantum dynamics on a quantum computer was based on product formulas.
Indeed, Lloyd used a simple product formula in the seminal work~\cite{lloyd96} to simulate quantum evolution under local Hamiltonians.
Later work~\cite{BAC+07} used high-order product formulas to simulate the broader class of sparse Hamiltonians~\cite{AT03}.
Since then, a host of new algorithmic techniques have been developed, enabling the design of quantum algorithms for simulating quantum dynamics with increasingly improved asymptotic performance as a function of various parameters such as the evolution time, the system size, and the allowed simulation error~\cite{BCC+15,BCK15,camp19,COS19,LC17,LC19,LW19,HW23,NBA24,CBH24}.
Despite the improvements in recent years, product formulas are still preferred for practical applications~\cite{CMN+18} and have been predominantly used in experimental implementations~\cite{BCC06,LHN+11,BLK+15}.

Product formulas approximate the time evolution $\exp(-\i Ht)$ generated by a time-independent Hamiltonian of the form $H=\sum_\ell H_\ell$ using a product of exponentials of the individual summands $H_\ell$.
For sufficiently small values of~$t$, the approximation error generally scales as $\order{t^{k+1}}$, where $k$ is the order of the product formula.
For longer times, the time evolution is divided into many steps, and short-time simulation is performed for each step. 
The number of steps needed to reach a certain error tolerance is reduced by increasing the order, which motivates using higher-order product formulas.
Suzuki~\cite{suz90} developed a systematic method to generate arbitrarily high-order formulas.
Suzuki's product formulas are typically used in quantum computing and are considered standard product formulas.
However, the number of exponentials needed to generate higher-order formulas using Suzuki's method grows rapidly with the order.
Yoshida~\cite{yos90} developed an alternative method to obtain product formulas with fewer exponentials. Unlike Suzuki's method, Yoshida's method does not yield an explicit analytic form for the higher-order formulas. Instead, it requires deriving and solving a complicated set of simultaneous nonlinear polynomial equations.
Due to these limitations, product formulas of low orders, typically of order $k\leq8$, have been mainly used in practice~\cite{CMN+18,CS19}.

Product formulas are preferred for practical application due to their simplicity and because they do not require any ancilla qubits or potentially costly operations, such as block encodings or controlled evolutions.
Additionally, the empirical error of product formulas is typically much better than the theoretical error bounds.
For these reasons, product formulas are expected to remain a competitive approach for Hamiltonian simulation in practical applications, particularly for noisy intermediate-scale and early fault-tolerant quantum computation. These considerations motivate developing approaches to improve the efficiency of product formulas.

Here we introduce corrected product formulas~(CPFs) that can improve the efficiency of existing product formulas by orders of magnitude.
CPFs are achieved by injecting auxiliary terms, which we call correctors, into existing product formulas.
We introduce three types of correctors to establish CPFs that reduce the approximation error of existing product formulas by orders of magnitude, thereby greatly improving the simulation cost of prior product formulas quantified by the number of exponentials used, which consequently leads to a reduction in the gate count scaling.
The corrected product formulas we establish are high-order product formulas that can be used for perturbed $(\alpha \ll1)$ and non-perturbed $(\alpha=1)$ systems with a Hamiltonian of the form $H=A+\alpha B$, where $A$ and $B$ have comparable norms and are exactly simulatable.
This Hamiltonian form is a common characteristic of lattice Hamiltonians used as effective models for many physical systems; such Hamiltonians typically can be divided into two exactly simulatable parts because either they contain pairwise commuting terms or they can be efficiently diagonalized.
We also discuss the application of CPFs for cases where the Hamiltonian partitions are not exactly simulatable.

For non-perturbed systems, we establish a CPF of order~$2k$~(CPF$2k$) with an error bounded as $\order{t^{2k+3}}$, which is two orders better than the error bound $\order{t^{2k+1}}$ for the standard $(2k)$th-order product formula~(PF$2k$).
We show CPFs are more advantageous for perturbed systems, as they allow the ratio of norms of the Hamiltonian partitions to be used as an auxiliary parameter to control the simulation error.
Specifically, the CPF$2k$ we establish for perturbed systems has an error bound $\order{\alpha^2 t^{2k+1}}$ that is a factor of $\alpha$ better than the error bound $\order{\alpha t^{2k+1}}$ of the standard PF$2k$.
We also develop correctors for the product formulas based on Yoshida's method.
The CPFs we establish for these formulas achieve the error bound $\order{\alpha^2t^{2k+1}}$ for perturbed systems, providing a factor of $\alpha$ improvement for their non-corrected versions.
Furthermore, we establish several customized low-order corrected product formulas that yield orders of magnitude improvements compared with the error bound of their non-corrected versions with respect to both~$\alpha$ and~$t$.
Similar to low-order standard formulas, the low-order CPFs we establish could be preferred for practical applications.

To utilize these correctors in practical applications, we provide procedures to decompose the correctors into a product of the exponential of Hamiltonian terms---a process we call \textit{compilation} for correctors---and rigorously analyze the compilation error.
The compilations we provide for correctors have applications that extend beyond CPFs. They can be used to simulate the time evolution of a linear combination of nested commutators~\cite{CW13,CCH22}, enabling efficient synthesis of complicated unitaries on a quantum simulator using a limited set of native gates.
We show the correctors only increase the cost of product formulas by a small additive or multiplicative factor while the reduction in the total cost of simulation is significant. Indeed, for some cases, the additional cost due to correctors is only an additive constant factor independent of the simulation time~$t$.
To validate our established error bounds and to compare the performance of CPFs against their non-corrected versions, we perform numerical simulations for various (non-)perturbed lattice Hamiltonians. We also implement both CPFs and standard PFs on actual quantum hardware, as well as noisy and noiseless hardware simulators.
Our numerical results show that the theoretical performance of CPFs matches or exceeds the \textit{empirical} performance of standard product formulas. 
Our hardware implementations indicate that CPFs can achieve higher accuracy simulations when compared to standard PFs in a \textit{practical} setting.

%==================================================
\subsection{Relation to prior work}
%==================================================

In Ref.~\cite{suzuki1995hybrid}, Suzuki developed several ``hybrid" product formulas of the fourth order using an approach similar to the corrected product formulas we establish in this work but limited to the fourth-order case. In particular, Suzuki constructed fourth-order product formulas by adding extra terms to some customized second-order product formulas, where the extra terms consist of nested commutators in~$A$ and~$B$. Suzuki did not present a compilation for these extra terms using the exponential of Hamiltonian terms.
Instead, the extra terms were simplified for a family of Hamiltonians comprised of a Laplacian operator and a potential energy operator that is a function of position.

Wisdom~\textit{et al.}~\cite{WHT96} developed a modified version of a commonly used second-order product formula known as the leapfrog integrator, given in~\cref{eq:PF2}, by adding corrector terms to the beginning and end of the leapfrog integrator. 
When applied to perturbed systems with a Hamiltonian of the form $H=A+\alpha B$, their modified product formula improves the approximation error of the leapfrog integrator from $\order{\alpha t^3}$ to $\order{\alpha^2 t^3 + \alpha t^{k+1}}$ for any integer $k$. However, this improvement is not explicitly mentioned in Ref.~\cite{WHT96}.
For such perturbed systems, Laskar and Robutel~\cite{LR01} showed a family of product formulas \textit{exists} that achieve the same error improvement. They also constructed explicit product formulas for $k\le 10$ by solving a set of nonlinear algebraic equations that become increasingly more complicated as~$k$ increases.
New families of product formulas similar to those of Laskar and Robutel have also been developed in Ref.~\cite{BCF+13}.
These prior works on improved second-order product formula for perturbed systems have been used in classical computing, particularly in astrophysics to simulate planetary systems.
See, e.g., Refs.~\cite{RT05,RTB19} and the references in the software package \texttt{REBOUND}~\cite{rebound} for simulating planetary systems.

Another line of work in the literature is based on the processing method for product formulas~\cite{BCR99,BCM06,BCM08}, in which a product formula is generated by the composition $P\Sigma P^{-1}$, where~$\Sigma$ is a kernel operator and~$P$ is a processor operator. This structure is especially useful in long-time simulations, as $P$ and $P^{-1}$ cancel in successive steps of the simulation. Ref.~\cite{BCR99} presents a general study of symplectic product formulas with processing and explicitly constructs processed product formulas of up to order~6. Later work~\cite{BCM06} extends this approach to formulas of order~12. Processed product formulas can be viewed as a generalization of the work of Wisdom~\textit{et al.}~\cite{WHT96}

In contrast, the corrected product formulas we establish in this work are high-order product formulas that apply to both perturbed $(\alpha \ll1)$ and non-perturbed $(\alpha=1)$ systems.
Using the three types of correctors we introduce, we also develop several customized low-order product formulas that reduce the error scaling of their non-corrected versions with respect to both~$\alpha$ and~$t$.
Furthermore, we provide constructive procedures for compiling the correctors using the Hamiltonian terms and rigorously analyze the error of the compilations.

Our established corrected product formulas achieve the error scaling $\order{\alpha^2 t^{k+1}}$ for perturbed systems.
As we were completing this work, we became aware of independent work in Ref.~\cite{BCD+24} on simulating perturbed systems using product formulas with the same error scaling as ours.
Specifically, Ref.~\cite{BCD+24} presents an algorithm called THRIFT for approximating $\exp(-\i Ht)$ that achieves the error $\order{\alpha^2 t^{k+1}}$ for perturbed systems with the Hamiltonian $H=A+\alpha B$.
The approach of Ref.~\cite{BCD+24} is to move into the interaction frame of~$A$ and simulate the resulting interaction-picture Hamiltonians using product formulas.
Hamiltonian simulation in the interaction picture~\cite{GN19} generally requires ancillary qubits and performing some controlled evolutions to achieve a gate cost that, in theory, scales better with the evolution time and simulation error than product formulas. However, empirical studies show product-formula-based approaches can perform better in practice~\cite{CMN+18,CS19}.
Reference~\cite{BCD+24}, however, avoids the requirements of the interaction-picture simulation and uses the structure of the Hamiltonian to provide an ancilla-free simulation that involves only a product of exponentials according to terms of the Hamiltonian.
Indeed, the product formula proposed in Ref.~\cite{BCD+24} involves exponentials of $A$ and $A+\alpha H_j$, where $H_j$ are summands of $B$; i.e., $B=\sum_j H_j$.
Reference~\cite{BCD+24} also shows the quadratic error scaling in~$\alpha$ cannot be improved using products of time-ordered evolutions according to the terms of the Hamiltonian. However, it presents non-product-formula approaches based on the Magnus expansion to achieve beyond-quadratic error scaling only for very small~$\alpha$.
Beyond-quadratic error scaling in~$\alpha$ using the Magnus expansion also appears in a concurrent work~\cite{ST24}.

While the exponential of $A+\alpha H_j$ can be efficiently constructed by quantum gates for some Hamiltonians of the form $H=A+\alpha B$ with exactly simulatable~$A$, particularly for typical lattice Hamiltonians as shown in Ref.~\cite{BCD+24}, constructing such exponentials by quantum gates could be challenging in general.
The corrected product formulas we establish in this work apply to both perturbed and non-perturbed systems and only use the exponentials of the Hamiltonian terms~$A$ and~$B$.
Additionally, the low-order CPFs we have established, specifically those in Tables~\ref{tab:correctors} and~\ref{tab:correctorsYPF} with a symplectic corrector, only increase the simulation cost based on their non-corrected versions by a negligible additive factor. In comparison, the simulation cost of the THRIFT algorithm using the proposed product formula in~\cite[Eq.~(13)]{BCD+24} is higher than that of the associated ordinary product formula~\cite[Eq.~(12)]{BCD+24} by a multiplicative factor.

%==================================================
\subsection{Notation}
%==================================================
Key notations we use are as follows.
We denote the nested commutator $[A,[B, [C,\ldots]]]$ by $[A,B,C,\ldots]$ and use the term `depth' to refer to the number of iterations in a nested commutator.
For instance, $[A,B]$ has depth one and $[A,A,A,B]=[A,[A,[A,B]]]]$ has depth three.
We also use the notation of adjoint action for nested commutators recursively defined with the base case $\text{ad}_A(B) := [A,B]$ as $\text{ad}^j_A(B) := \text{ad}^{j-1}_A([A,B])$ for any integer $j>1$, the depth of the nested commutator.
One example of this notation is $\text{ad}^3_A(B) = [A,A,A,B]$.
We denote the spectral norm of operators by~$\norm{\cdot}$.
Unless otherwise specified, $\log$ is used for the natural~logarithm.

%==================================================
\subsection{Organization}
%==================================================
The rest of this paper is organized as follows.
We begin with a high-level description of our approach to establish corrected product formulas and an overview of our main results in \cref{sec:overview}.
In \cref{sec:correctors}, we develop several customized correctors for standard product formulas of low orders and elaborate our established correctors for high-order standard product formulas applicable to both perturbed and non-perturbed systems.
There, we also discuss the advantages of correctors for a broad class of structured systems.
\cref{sec:yoshidaPFs} covers the correctors we establish for the product formulas based on Yoshida's method.
We elaborate in \cref{sec:compile} our approach for compiling the correctors and present compilations for various correctors we develop.
Our numerical simulations for comparing the performance of corrected and non-corrected product formulas are presented in \cref{sec:numerics}. We cover our hardware implementations in \cref{sec:hardware} and conclude by discussing our results in \cref{sec:discussion}.

%==================================================
\section{Corrected product formulas:
Approach and main results}
\label{sec:overview}
%==================================================
This section provides an overview of corrected product formulas (CPFs) and our main results.
Technical details to establish the results are provided in subsequent sections.
We begin with the setup and assumptions.
Consider a system that evolves under a Hamiltonian $H=\sum_j H_j$, a sum of time-independent terms~$H_j$.
We assume $H$ can be divided into two parts as $H=A+B$, such that both $A$ and $B$ parts have similar norms and can be exactly simulated.
Many physically relevant systems are described by a Hamiltonian that satisfies the norm and exact simulatability assumptions. Examples of such Hamiltonians are provided in our numerical studies in \cref{sec:numerics}.
For perturbed systems where the norm of one partition is significantly smaller than the other, we express the Hamiltonian as $H=A+\alpha B$ with $0<\alpha\ll1$, which we call the perturbation parameter, and where $A$ and $B$ have similar norms.
We will later discuss the possibility of relaxing the exact simulatability assumption for the small-norm partition for perturbed systems, i.e., the partition $\alpha B$. We will also argue how our approach could be effective for generic cases where Hamiltonian partitions are not exactly simulatable.

To develop CPFs, we focus on the Hamiltonian~$H=A+B$ and replace $B$ with $\alpha B$ for perturbed systems. 
Let $\lambda$ be a complex number.
Product formulas approximate $\exp(\lambda H)$ by a
product of exponentials of~$A$ and~$B$.
For Hamiltonian simulation, we have $\lambda=-\i t$ with $t$ the simulation time.
The well-known Lie product formula
\begin{equation}
\label{eq:PF1}
    S_1(\lambda) := \e^{\lambda A}\e^{\lambda B}
\end{equation}
is a first-order product formula (PF1) that approximates the exponential $\exp(\lambda H)$ to the first order in $\lambda$.
The second-order product formula (PF2)
\begin{equation}
\label{eq:PF2}
    S_2(\lambda) := \e^{\lambda A/2}\e^{\lambda B} \e^{\lambda A/2}
\end{equation}
improves the approximation to the second order in $\lambda$.
Using $S_2(\lambda)$ as the base case, Suzuki constructed a $(2k)$th-order product formula (PF$2k$) defined recursively as~\cite{suz90}
\begin{equation}
\label{eq:PF2k}
    S_{2k}(\lambda):=
    [S_{2k-2}(p_k\lambda)]^2
    S_{2k-2}((1-4p_k)\lambda)
    [S_{2k-2}(p_k\lambda)]^2
    \quad
    \text{with}
    \quad
    p_k:=\frac{1}{4-4^{1/(2k-1)}},
\end{equation}
which progressively improves the approximation by increasing the order parameter $2k$.
We refer to these formulas as the \textit{standard product formulas}.
PF$2k$ can be expressed as $S_{2k}(\lambda)= \exp(K_{2k})$
for some operator~$K_{2k}$ that we call the \textit{kernel} of the product formula.
Product formulas thus generate an approximation of the exact kernel~$\lambda H$, and the approximate kernel improves by increasing~$k$.
For clarity, we note that our terminology differs slightly from that used for processed product formulas~\cite{BCM06}, in which the exponential of~$K$ is referred to as the kernel.

We use correctors to improve the approximation quality of the standard product formulas without increasing~$k$.
Correctors are auxiliary terms of the form $\exp(\pm C)$ multiplied to the left and right of the standard product formulas to reduce the error in approximating the exact kernel.
Hereafter, we refer to $C$ as the corrector for convenience.
We use three types of correctors described below.
\begin{itemize}
    \item \textbf{Symplectic corrector}
    \begin{equation}
    \label{eq:sympC}
        \e^C \e^{K} \e^{-C} = \e^{K'}
        \quad
        \text{with}
        \quad
        K' :=  \e^{\ad_C} K= K + \ad_C(K) + \tfrac12 \ad^2_C(K) + \cdots
        \qquad
        \text{(see \cref{lem:keyformula}).}
    \end{equation}
    \item  \textbf{Symmetric corrector}
    \begin{equation}
    \label{eq:symC}
        \e^{C}\e^{K}\e^{C} = \e^{K'}
        \quad
        \text{with}
        \quad
        K' := K+2C-\tfrac16[C+K,C,K]+\cdots,
    \end{equation}
    where ``$\cdots$" comprises nested commutators with a higher depth.
    \item \textbf{Composite corrector}: Any composition of the previous two correctors.
    For the composite corrector denoted as $C_2\circ C_1$, the corrector~$C_2$ is applied after~$C_1$.
    For example, if~$C_1$ is symmetric and~$C_2$ is symplectic $C_2\circ C_1$ denotes the transformation $\e^{C_2}\e^{C_1}\e^K\e^{C_1}\e^{-C_2}$.
\end{itemize}
Of note is the symplectic corrector that adds only a negligible additive cost to that of standard product formulas.
Specifically, for $r$ steps of simulation, we have
\begin{equation}
    \left(\e^C S_{2k}\left(\lambda/r\right)\e^{-C} \right)^r
    = \e^C S_{2k}\left(\lambda/r\right)^r \e^{-C},
\end{equation}
so only the implementation cost of $\exp(\pm C)$ in one step is added to that of standard product formulas.

We use the expression for the modified kernel $K'$ in \cref{eq:sympC} and~\cref{eq:symC} to construct several customized correctors for PF1 and PF2.
The correctors and their effect on standard product formulas are summarized in \cref{tab:correctors}.
The correctors we develop apply to non-perturbed systems ($\alpha=1$) and perturbed systems ($\alpha\ll 1$) with Hamiltonian of the form $H=A+\alpha B$.
The key to establishing a corrector for perturbed systems is an analytic expression for the kernel $K_2$ of PF2 that comprises all first-order error terms in $\alpha$: the error terms in $K_2$ that are large in magnitude.
We use this analytic expression to construct a symplectic corrector that removes all of the large error terms in~$K_2$, yielding a CPF2 with the leading error $\order{\alpha^2|\lambda|^3}$, which is a factor of $\alpha$ better than the leading error of the standard PF2.
Specifically, we use the following expression to obtain the large errors in the kernel of PF2 (See Proposition~\ref{prop:largeterms} for a generic case)
\begin{equation}
\label{eq:K2}
    \e^{A/2}\e^{B}\e^{A/2} = \e^{K_2} 
    \quad
    \text{with}
    \quad
    K_2 \equiv_{(\geq2)} A + B + \sum_{j=1}^{\infty} \frac{B_{2j}(1/2)}{(2j)!} \ad^{2j}_A(B),
\end{equation}
where $\equiv_{(\geq2)}$ denotes equality modulo terms with degree $\geq 2$ in $B$, and where $B_n(x)$ are Bernoulli polynomials defined as
\begin{equation}
\label{eq:bernoli}
    \frac{t\e^{tx}}{\e^t -1} = \sum_{n=0}^{\infty} B_n(x) \frac{t^n}{n!};
\end{equation}
see \cref{eq:bernoulli} for a few nonzero Bernoulli polynomials at $x=1/2$.

Equation~(\ref{eq:K2}) allows us to identify the large (in magnitude) error terms in the kernel of PF2 when simulating perturbed systems.
In particular, by replacing $A\to \lambda A$ and $B\to \alpha \lambda B$, the terms in the summation in \cref{eq:K2} are the large error terms in the kernel of PF2 for perturbed systems.
The first $k$ error terms, for any $k\geq 1$, in this summation can be removed by a symplectic corrector that can be constructed using the modified kernel $K'$ in \cref{eq:sympC}.
The symplectic corrector yields a $k$th-order CPF2 (not to be confused with the order of standard product formulas).
We provide the symplectic corrector and its effect on PF2 in the following proposition.

\begin{proposition}[High-order symplectic CPF2 for perturbed systems]
\label{prop:high-order-CPF2}
    Given a Hamiltonian $H = A + \alpha B$ with the perturbation parameter $0<\alpha\ll 1$,
    and a complex number $\lambda$ with $|\lambda|(\norm{A}+\alpha\norm{B})<\frac{\log{2}}2$.
    Let
    \begin{equation}
    \label{eq:PF2_pert}
        S_2(\lambda):=
        \e^{\lambda A/2}
        \e^{\alpha \lambda B/2}
        \e^{\lambda A/2}
    \end{equation}
    be the second-order standard PF.
    For any integer $k\geq 1$, the second-order CPF defined as
\begin{equation}
    \label{eq:CPF2symp}
    S^c_2(\lambda):=\e^{C(k)} S_2(\lambda) \e^{-C(k)}
    \quad
    \text{with the symplectic corrector}
    \quad
    C(k) := \alpha \sum_{j=1}^{k} \frac{B_{2j}(1/2)}{(2j)!} \lambda^{2j}\,\ad^{2j-1}_A(B)
\end{equation}
approximates $\exp\left(\lambda H\right)$
with error $\norm{\exp\left(\lambda H\right) - S^c_2(\lambda)} \in \order{\alpha^2|\lambda|^3 + \alpha |\lambda|^{2k+3}}$.
\end{proposition}

We remark that care must be taken in using the above proposition. As the error scales with both $\alpha$ and~$|\lambda|$, the error with respect to both of them needs to be considered.
When the perturbation strength~$\alpha$ is very small, fewer terms in the series for~$C(k)$, perhaps much smaller than $k$ terms in the summation, may suffice to achieve a given target accuracy.
We also note that the constraint on~$|\lambda|$ ensures that both the series expansion of the product formula and the series defining the corrector are convergent
(see \cref{apx:proofs}).

The established symplectic corrector removes first-order error terms in $\alpha$ in \textit{all} orders of $\lambda$, up to the $2k$th order.
We utilize this feature of the corrector and that a CPF2 with a symplectic corrector has time-reversal symmetry, in the sense that $S^c_2(\lambda) S^c_2(-\lambda)=\id$, to construct high-order CPFs by a recursive formula \'a la Suzuki~\cite{suz90}.
We state the recursive formula and the improvement offered by correctors in the following~theorem.

\begin{theorem}[High-order CPFs for perturbed systems]
    The $(2k)$th-order CPF defined recursively~as
    \begin{equation}
    \label{eq:CPF2k_pert}
    S^c_{2k}(\lambda) :=
    \left[S^c_{2k-2} (p_k \lambda)\right]^2
    S^c_{2k-2} ( (1-4p_k) \lambda)
    \left[S^c_{2k-2} (p_k \lambda)\right]^2
    \quad (\forall k\geq2)
\end{equation}
with the base case $S^c_2(\lambda)$ given in \cref{eq:CPF2symp} and $p_k$ in \cref{eq:PF2k} approximates $\exp\left(\lambda H\right)$ with error $\order{\alpha^2 |\lambda|^{2k+1}}$.
\end{theorem}

The high-order CPFs in this theorem are built from a CPF2 with a symplectic corrector. This approach reduces the approximation error by a factor of $\alpha$. Still, it also introduces additional terms in product formulas that are not canceled in consecutive simulation steps, resulting in a multiplicative factor in the simulation cost. The multiplicative factor is small but can be avoided by constructing CPFs with a symplectic corrector. Such correctors only add a negligible additive cost to the total simulation cost. We construct a CPF4 with a symplectic corrector using \cref{eq:K2} and by analyzing the error terms of the standard PF4.
We present the CPF4 with a symplectic corrector in the following proposition.
Proof is provided in \cref{subsubsec:CPFs_pert_sys}.

\begin{proposition}[Symplectic CPF4 for perturbed systems]
\label{prop:CPF4}
    Define the constants
    \begin{equation}
    \label{eq:constants}
        s:=\frac1{4-\sqrt[3]{4}}
        \quad
        \text{and}
        \quad
        c := \frac7{5760}\left(4s^5+(1-4s)^5\right)
        +\frac1{72} s (1-2s)(1-3s)(1-4s)(1-5s)
    \end{equation}
    and let the fourth-order PF be
    \begin{equation}
    \label{eq:PF4}
        S_4(\lambda) :=
        [S_2(s\lambda)]^2
        S_2((1-4s)\lambda)
        [S_2(s\lambda)]^2
    \end{equation}
    with $S_2(\lambda)$ given in \cref{eq:PF2_pert}.
    Then, the fourth-order CPF defined as
    \begin{equation}
    \label{eq:CPF4symp}
        S^c_4(\lambda):= \e^{C} S_4(\lambda) \e^{-C} 
        \quad
        \text{with the symplectic corrector}
        \quad
        C = c \lambda^4 \ad^3_A(\alpha B)
    \end{equation}
approximates $\exp\left(\lambda H\right)$ with the error $\order{\alpha^2|\lambda|^5 + \alpha |\lambda|^7}$.
\end{proposition}

We take a similar approach to establishing high-order CPFs for a non-perturbed system.
That is to say, first we construct a CPF2 with time-reversal symmetry and then use it as the base case to construct higher-order CPFs.
The CPF2 we construct for non-perturbed systems generates a kernel with the leading error $\order{|\lambda|^5}$ that is two orders of magnitude better than that for the standard PF2, which has the leading error $\order{|\lambda|^3}$.
The time-reversal symmetry of the constructed CPF2 allows it to be used as the base case to recursively construct high-order CPFs that provide two orders of magnitude improvement for the error of standard PF with the same order as stated in the following theorem.
\begin{theorem}[Higher-order CPFs for non-perturbed systems]
    Let $S^c_2(\lambda)$ be a CPF2 with the time-reversal symmetry that approximates $\exp\left(\lambda H\right)$ with error $\order{|\lambda|^5}$ for a system with the Hamiltonian $H=A+B$.
    Then the $(2k)$th-order CPF defined recursively~as
    \begin{equation}
    \label{eq:CPF2k_nonpert}
    S^c_{2k}(\lambda) :=
    \left[S^c_{2k-2} (a_k \lambda)\right]^2
    S^c_{2k-2} ( (1-4a_k) \lambda)
    \left[S^c_{2k-2} (a_k \lambda)\right]^2
    \quad
    \text{with}
    \quad
    a_k = \frac{1}{4-4^{1/(2k+1)}},
\end{equation}
has time-reversal symmetry as well and approximates $\exp\left(\lambda H\right)$ with an error scaling as $\order{|\lambda|^{2k+3}}$.
\end{theorem}
\noindent
Notice that the parameter $a_k$ differs from $p_k$ in Suzuki's recursive formula in \cref{eq:PF2k}.

The correctors we establish for (non-)perturbed systems are in terms of a linear combination of nested commutators; see, e.g., \cref{eq:CPF2symp}.
To utilize the established correctors in practical applications, we provide a compilation for $\exp(\pm C)$ in terms of a product of the exponential of the  Hamiltonian terms.
More specifically, we provide a decomposition for $\exp(C)$ as
$\prod_{j} \e^{a_j\lambda A}\e^{b_j\lambda B} = \exp{(C + \textsc{error})}$ for some appropriately chosen real coefficients $a_j$ and $b_j$.
This compilation incurs an error, denoted by \textsc{error}, but we keep it smaller than or within the same order as the error in the corrected product formula. \cref{tab:compilations} summarizes compilations for various correctors and their associated errors and costs.

{\small
\begin{table}
\renewcommand{\arraystretch}{1.04}
\centering
  \begin{tabular}{ll|ll}
    \toprule
    Product& Error bound for&\multirow{2}{*}{Correctors} &Error bound for\\
    formula& non-corrected PF& & corrected PF\\
    \midrule
    \multirow{5}{*}{PF1} & \multirow{5}{*}{$\order{\alpha|\lambda|^2}$} & $C_\text{symp} = \tfrac\lambda2 \alpha B$ & $\order{\alpha|\lambda|^3}$\\
    &&$C_\text{symp} = \tfrac\lambda2 \alpha B + \tfrac{\lambda^2}{12}\ad_A(\alpha B)$ & $\order{\alpha^2|\lambda|^3 + \alpha|\lambda|^4}$ \\
    &&$C_\text{sym} = -\tfrac{\lambda^2}{4}\ad_A(\alpha B) - \tfrac{\lambda^3}{12}\ad^2_{\alpha B}(A)$ &$\order{\alpha|\lambda|^3}$\\
    && $C_\text{com} = C_\text{symp}\circ C_\text{sym}$
    \text{with} $C_\text{symp}= \tfrac{\lambda^2}{12}\ad_A(\alpha B)$ & $\order{\alpha|\lambda|^4}$\\
    \midrule
    \multirow{4}{*}{PF2} &\multirow{4}{*}{$\order{\alpha|\lambda|^3}$} &$C_\text{symp}=-\tfrac{\lambda^2}{24} \ad_A(\alpha B)$ &$\order{\alpha^2|\lambda|^3 + \alpha|\lambda|^5}$\\
    && $C_\text{sym}=\tfrac{\lambda^3}{48}\ad^2_A(\alpha B)-\tfrac{\lambda^3}{24}\ad^2_{\alpha B}(A)$ &$\order{\alpha|\lambda|^5}$\\
    && $C_\text{com}= C_\text{symp}\circ C_\text{sym}$
    with $C_\text{sym}=-\tfrac{\lambda^3}{48}\ad^2_{\alpha B}(A)$
    &$\order{\alpha|\lambda|^5}$\\
    &&$C_\text{symp}= \displaystyle \sum_{j=1}^{k} \frac{B_{2j}(\tfrac12)}{(2j)!} \lambda^{2j}\ad^{2j-1}_A(\alpha B)$
    &$\order{\alpha^2|\lambda|^3 + \alpha |\lambda|^{2k+3}}$\\
    \midrule
    PF4 & $\order{\alpha|\lambda|^5}$
    &$C_\text{symp}=c\lambda^4 \ad^3_A(\alpha B)$ with $c$ given in \cref{eq:constants}.
    &$\order{\alpha^2|\lambda|^5}$\\
    \midrule
    &\multirow{6}{*}{$\order{\alpha|\lambda|^{2k+1}}$}
    &$C_\text{symp}$ in the last line of PF2 correctors used in
    &\multirow{2}{*}{$\order{\alpha^2|\lambda|^{2k+1}}$}\\
    &&the base case of CPF$2k$;
    see \cref{eq:CPF2k_pert}.&\\
    PF$2k$ && &\\
    $\forall k\ge 2$ && $C_\text{sym}$ in PF2 correctors used in the base case of
    &\multirow{2}{*}{$\order{\alpha|\lambda|^{2k+3}}$}\\
    &&CPF$2k$; see \cref{eq:CPF2k_nonpert}.
    &\\
    \bottomrule
  \end{tabular}
  \caption{\label{tab:correctors}
  Various correctors and the error bounds of (non-)corrected product formulas for perturbed ($\alpha\ll 1$) and non-perturbed ($\alpha=1$) systems with a Hamiltonian of the form $H=A+\alpha B$, where partitions $A$ and~$B$ have comparable norms.
  Here $\text{ad}_A(B) := [A,B]$ denotes the adjoint action and $\text{ad}^j_A(B) = \text{ad}^{j-1}_A([A,B])$.
  Observe that some correctors are ineffective for non-perturbed systems.}
\end{table}
}

%==================================================
\section{Correctors for standard product formulas}
\label{sec:correctors}
%==================================================

In this section, we develop various correctors for the standard product formulas.
To show how correctors can improve quantum simulation based on product formulas, we consider a generic Hamiltonian of the form $H=A+B$ and replace $B\to \alpha B$ for perturbed systems.
We begin by developing correctors for the first- and second-order product formulas.
Then, we describe the effect of correctors for perturbed and some structured systems.
We finish this section by developing correctors for higher-order standard product formulas.

%==================================================
\subsection{Correctors for PF1}
\label{subsec:CPF1}
%==================================================

We begin with correctors for the first-order formula in \cref{eq:PF1}.
For a $\lambda\in \mathbb{C}$ with $|\lambda|\leq 1$,
the kernel of PF1 is
\begin{equation}
\label{eq:PF1_kernel}
    K_1:=\log S_1(\lambda) = \lambda H +\frac{1}{2}\lambda^2[A,B]
    +\frac{1}{12}\lambda^3[A-B,A,B]
    + \order{|\lambda|^4},
\end{equation}
where the leading error is second-order in~$\lambda$.
Let us take
$C=\lambda B/2$ as a symplectic corrector.
Evidently, this corrector maps PF1 to PF2 because $\exp(C)S_1(\lambda)\exp(-C)=S_2(\lambda)$.
PF1 is thus as effective as PF2 with a negligible additive cost.
That is, for a simulation with $r$ steps we have
\begin{equation}
    S_2(\lambda)^r =
    (\e^C S_1(\lambda) \e^{-C})^r =
    \e^CS_1(\lambda)^r\e^{-C} =
    \e^{\frac\lambda2 B}
    S_1(\lambda)^r
    \e^{-\frac\lambda2 B}.
\end{equation}
The additional cost here is due to the exponentials at the beginning and end of the simulation.

Let us now take $C= \lambda B/2 + c_2 \lambda^2 [A,B]$ as a symplectic corrector with $c_2$ a constant to be identified.
By \cref{eq:sympC} and $K_1$ in \cref{eq:PF1_kernel},
we have $\e^C S_1(\lambda) \e^{-C} = \e^{K'}$ with
\begin{align}
    K' &= K_1 + [C,K_1] + \tfrac12[C,C,K_1] + \order{|\lambda|^4}\\ \label{eq:PF1W}
    &= \lambda H +(\tfrac1{12}-c_2)\lambda^3[A,A,B]+
    (c_2-\tfrac1{24})\lambda^3
    [B,B,A] + \order{|\lambda|^4}.
\end{align}
The proof is given in \cref{apx:proofs}.
Setting $c_2=1/12$ removes the second term here, and we obtain the CPF1
\begin{equation}
\label{eq:PF1Csymp2}
    \e^{C_\text{symp}} S_1(\lambda) \e^{-C_\text{symp}} = \e^{\lambda H + \frac{\lambda^3}{24}[B,B,A]+\order{|\lambda|^4}}
    \quad \text{with} \quad
    C_\text{symp} = \tfrac\lambda2 B + \tfrac{\lambda^2}{12}\ad_A(B).
\end{equation}
This CPF1 is particularly useful for simulating perturbed systems with Hamiltonian $H=A+\alpha B$ ($\alpha\ll 1$), as its error scales as $\order{\alpha^2|\lambda|^3}$.
With this CPF, $\alpha$ now serves as a parameter that can be used to reduce the approximation error.
We defer a more detailed discussion of correctors for perturbed systems to~\cref{subsec:pert_sys}.

We now show how injecting correctors with the symmetric approach can remove the second- and third-order error terms.
Let
$C = c_2\lambda^2 [A,B] + c_3 \lambda^3 [B,A,B]$ be a symmetric corrector with constants $c_2$ and $c_3$ to be identified.
By \cref{eq:symC}, this corrector modifies the PF1's kernel $K_1$ in \cref{eq:PF1_kernel} to $K'_1$, i.e., $\e^C S_1(\lambda) \e^C = \e^{K'_1}$, as
\begin{align}
    K'_1= K_1+2C+\order{|\lambda|^4}
    &=\lambda H+ (2c_2+\tfrac12)\lambda^2[A,B]
    +\tfrac{1}{12}\lambda^3[A-B+24 c_3 B,A,B] + \order{|\lambda|^4}\\
    &= \lambda H + \tfrac{1}{12}\lambda^3[H, A,B]
    +\order{|\lambda|^4},
\end{align}
where we set $c_2=-1/4$ and $c_3=1/12$.
Hence we have the symmetric CPF1
\begin{equation}
\label{eq:PF1Csym}
    \e^{C_\text{sym}} S_1(\lambda) \e^{C_\text{sym}} = \e^{\lambda H + \tfrac{1}{12}\lambda^3[H, A,B]
    +\order{|\lambda|^4}}
    \quad \text{with}\quad
    C_\text{sym}=-\frac14\lambda^2\ad_A(B)-\frac1{12}\lambda^3\ad^2_B(A).
\end{equation}
The leading error term here can be removed by symplectic corrector.
Specifically, applying the symplectic corrector $C_\text{symp}=\frac{\lambda^2}{12}\ad_A(B)$ after the symmetric corrector yields a composite corrector that removes both the second- and third-order error terms from the kernel of PF1. The resulting kernel is then $K'_1=\lambda H +\order{|\lambda|^4}$, so the leading error is of the fourth error.

%==================================================
\subsection{Correctors for PF2}
\label{subsec:CPF2}
%==================================================
The PF2 in \cref{eq:PF2} has time-reversal symmetry in the sense that $S_2(\lambda)S_2(-\lambda)=\id$.
This symmetry is crucial in developing a high-order product formula, and we preserve it in developing correctors for PF2.
We begin with a symplectic corrector.
The kernel of PF2 is
\begin{equation}
\label{eq:PF2_kernel}
    K_2 := \log S_2(\lambda) = \lambda H
    - \frac{\lambda^3}{24}[A+2B,A,B]
    +\sum_{j\geq2} \lambda^{2j+1}E_{2j+1},
\end{equation}
where $E_j$ is the order-$\lambda^j$ error operator comprised of nested commutators of depth $j-1$.
Note that error terms have odd powers because of the time-reversal symmetry.
Let
\begin{equation}
\label{eq:Csymp_PF2}
    C = \frac{\lambda^2}{2}B_2\left(\frac12\right) \ad_A(B) =
    -\frac{\lambda^2}{24}[A,B]
\end{equation}
be a symplectic corrector;
this corrector is indeed a particular case of $C(k)$ in \cref{eq:CPF2symp} with $k=1$ that applies to both perturbed and non-perturbed systems.
By \cref{eq:sympC}, the CPF2 by this corrector is
\begin{equation}
\label{eq:CPF2symp1}
    S^c_2(\lambda) = \e^C S_2(\lambda) \e^{-C} =
    \e^{K_2+[C,K_2]+\order{|\lambda|^5}}=
    \e^{\lambda H + \frac{\lambda^3}{24}[B,B,A]+\order{|\lambda|^5}}
\end{equation}
as $[C,K_2]=\tfrac{\lambda^3}{24}[A+B,A,B]+ \order{|\lambda|^5}$.
The leading error of CPF2 here is identical to that of CPF1 in \cref{eq:PF1Csymp2}.
However, the next error for CPF2 is of fifth order, whereas it is of fourth order for CPF1.
More importantly, CPF2 has time-reversal symmetry, $S^c_2(\lambda)S^c_2(-\lambda)=\id$, which is a key feature for developing higher-order CPFs.

By an additional symmetric corrector $C_\text{sym}=\tfrac{\lambda^3}{48}[B,B,A]$,
we can remove the remaining third-order error term. 
This term can also be removed by first applying the symmetric corrector followed by the symplectic corrector.
Specifically, we have
\begin{equation}
\label{eq:CPF2comp}
    \e^C \e^{C_\text{sym}}
    S_2(\lambda) \e^{C_\text{sym}}\e^{-C} =
    \e^{\lambda H+\order{|\lambda|^5}}
    \quad\text{with}\quad
    C=-\frac{\lambda^2}{24}\ad_A(B),
    \;
    C_\text{sym}=-\frac{\lambda^3}{48}\ad^2_B(A).
\end{equation}
The benefit of this composite corrector is that its symplectic part is canceled in consecutive Trotter steps.
The CPF2 by this composite corrector has time-reversal symmetry.

The third-order error term in the kernel $K_2$ in \cref{eq:PF2_kernel} can be eliminated by only a symmetric corrector as
\begin{equation}
\label{eq:Csym_base}
    \e^{C_\text{sym}} S_2(\lambda) \e^{-C_{\text{sym}}}
    = \e^{\lambda H+\order{|\lambda|^5}}
    \quad\text{with}\quad
    C_\text{sym} = \frac{\lambda^3}{48}[A+2B,A,B],
\end{equation}
which follows from \cref{eq:symC}.
The corrector contains only odd-order terms, which enables constructing a compilation that is itself time-reversal symmetric (see \cref{subsec:compile_PF12}). 
This symmetry of compilation is required for the recursive construction of high-order CPFs. 

A summary of the correctors we established is provided in \cref{tab:correctors}.

%==================================================
\subsection{Correctors for perturbed and structured systems}
\label{subsec:pert_sys}
%==================================================
The correctors we developed in previous sections apply to generic systems.
Here, we show that correctors are more advantageous for perturbed systems in which one partition of the Hamiltonian has a small norm.
Specifically, we show how symplectic correctors enable using the perturbation parameter~$\alpha$ of such systems to reduce the simulation error.
Further, we discuss the advantages of correctors for some structured systems.

To demonstrate the advantage of correctors for perturbed systems, let us write the standard PF2 as
\begin{equation}
\label{eq:PF2new}
    S_2(\lambda) =\exp{\left(
    \lambda(A+\alpha B) +
    \sum_{j=1}^\infty \lambda^{2j+1} E_{2j+1}\right)
    },
\end{equation}
where $E_j$ is the error operator comprised of nested commutators of depth $j-1$ associated with the error term of order $\lambda^j$.
The nested commutator $E_3$, for instance, is $E_3 = -\tfrac1{24}[A,A,\alpha B]+\tfrac1{12}[\alpha B,\alpha B,A]$ as per \cref{eq:PF2}.
The leading error term of PF2 is $\order{\alpha |\lambda|^3}$ because the largest (in magnitude) term of $E_3$ is $\ad^2_A(\alpha B)=[A,A,\alpha B]$.
The symplectic corrector given in \cref{eq:CPF2symp1} indeed removes this term from $E_3$, enabling a CPF2 with the leading error $\order{\alpha^2|\lambda|^3}$.
Specifically, by replacing $B \to \alpha B$ in \cref{eq:Csymp_PF2} and~\cref{eq:CPF2symp1}, we have
\begin{equation}
\label{eq:CPF2_pert_order2}
    \e^C S_2(\lambda) \e^{-C}=
    \e^{\lambda(A+\alpha B) + \frac{\alpha^2\lambda^3}{24}[B,B,A]+\order{\alpha|\lambda|^5}}
    \quad
    \text{with}
    \quad
    C = -\alpha\frac{\lambda^2}{24} \ad_A(B)
\end{equation}
with the leading error $\order{\alpha^2|\lambda|^3}$, which improves the error $\order{\alpha|\lambda|^3}$ of standard PF2 by a factor of $\alpha$.

The error can be reduced further by designing a corrector that removes the term with the largest magnitude in $E_{2j+1}$ for larger values of $j$ as well.
To this end, we use the key observation that $\ad^{2j}_A(\alpha B)$ is the term with the largest magnitude in $E_{2j+1}$ and that the constant prefactor of such a term can be identified for each~$j$.
For example, the corrector in \cref{eq:CPF2_pert_order2}, which removes the error term $\order{\alpha|\lambda|^3}$, comprises the largest term $\ad_A(B)$ in $E_1$ with the prefactor $-\frac1{24}$ of the largest term $\ad^2_A(B)$ in $E_3$.
The corrector
$C=-\alpha\tfrac{\lambda^2}{24}\ad_A(B)+
\alpha\tfrac{7\lambda^4}{5760}\ad^3_A(B)$ not only removes the error $\order{\alpha|\lambda|^3}$ but also removes the error $\order{\alpha|\lambda|^5}$.
The second term of this corrector is the largest term $\ad^3_A(B)$ in $E_4$ with the prefactor $\frac{7}{5760}$ of the largest term $\ad^4_A(B)$ in $E_5$.

The leading error can be progressively improved by adding more terms to the corrector with appropriate constant prefactors.
We invoke the following proposition from Ref.~\cite[Proposition~1]{LR01} for the prefactors.  
\begin{proposition}
\label{prop:largeterms}
For any $s \in \R$, the following equation holds
    \begin{equation}
        \e^{s A}\e^{B}\e^{(1-s)A} = \e^{K} 
        \quad
        \text{with}
        \quad
        K \equiv_{(\geq2)} A + B + \sum_{j=1}^{\infty} \frac{B_j(s)}{j!} \ad^j_A(B),
    \end{equation}
where $\equiv_{(\geq2)}$ denotes equality modulo terms with degree $\geq 2$ in $B$ and $B_n(x)$ are Bernoulli polynomial in \cref{eq:bernoli}.
\end{proposition}

We remark that $s=1/2$ for PF2 and that $B_j(1/2)=0$ for all odd $j$.
By these remarks and invoking the above proposition, we obtain
\begin{equation}
\label{eq:CPF2_pert_order2k}
    \e^{C(k)} S_2(\lambda) \e^{-C(k)}=
    \e^{\lambda(A+\alpha B)
    +\order{\alpha^2|\lambda|^3 + \alpha \lambda^{2k+3}}
    }
    \quad
    \text{with}
    \quad
    C(k) = \alpha \sum_{j=1}^{k} \frac{B_{2j}(1/2)}{(2j)!} \lambda^{2j}\,\ad^{2j-1}_A(B).
\end{equation}
Proof is provided in \cref{apx:proofs}.
Note that this symplectic corrector removes all errors $\order{\alpha|\lambda|^{2j+1}}$ for $j\leq k$ from errors of~PF2.
That is, errors that are first order in $\alpha$ are removed, and those that are second order in $\alpha$ remain.
We use this fact in \cref{subsubsec:CPFs_pert_sys} to develop high-order CPFs for perturbed systems.

Perturbed systems are a type of structured systems where the structure is on the distribution of the norms of local terms in the Hamiltonian.
While correctors are advantageous for perturbed systems, we now discuss the advantage of correctors for a broader class of structured systems, where the structure is on commutators.
To this end, note that the leading error term in PF2 is comprised of two commutators, $\ad^2_A(B)$ and $\ad^2_B(A)$, and the symplectic correctors discussed so far only remove $\ad^2_A(B)$.
For structured systems where $B$ commutes with $[A,B]$, the leading error of CPF2 with symplectic correctors would be $\order{|\lambda|^5}$ and if the system is perturbed as well, the leading error would be $\order{\alpha^2|\lambda|^5}$.

The Hamiltonian for a broad class of systems is of the form $H=T+V(x)$, i.e., $A=T$ and $B=V(x)$, with~$T$ the kinetic part that is quadratic in the momentum~$p$ and~$V(x)$ the potential part that only depends on the position~$x$.
Example Hamiltonians include the Hamiltonian of a system of coupled harmonic oscillators or the Hamiltonian of massive quantum field theories~\cite{BSB+22}.
The commutator $\ad^2_B(A)$ for such systems can be written as some operator-valued function $f(x)$ that depends only on the position; therefore, it can be directly exponentiated.
Then we have
\begin{equation}
    \e^{C_\text{symp}}\e^{C_\text{sym}} S_2(\lambda)
    \e^{ C_\text{sym}} \e^{-C_\text{symp}}
    = \e^{\lambda H +\order{|\lambda|^5}}
    \quad
    \text{with}
    \quad
    C_\text{sym} = -\tfrac{\lambda^3}{48}f(x)
    \quad
    \text{and}
    \quad
    C_\text{symp} = C(k).
\end{equation}

We have $f(x)=- c^2\id$ when $T=p^2/2$ and $B=cx$ with~$c\in \mathbb{R}$, because
$\ad^2_B(A) = \frac12c^2[x,[x,p^2]]
%=\frac12c^2 [x,2\i p]
= -\i c^2 \id$ by $[x,p]=\i\id$.
In general, we have $f(x)=- \abs{\nabla V(x)}^2\id$ for $T=-\nabla^2/2$~\cite{suzuki1995hybrid}, so $f(x)$ can be directly exponentiated.

%==================================================
\subsection{Correctors for higher-order PFs}
\label{subsec:higher_order_correctors}
%==================================================
We now discuss how correctors improve the error of standard high-order product formulas applied to both non-perturbed and perturbed systems.
For non-perturbed systems, we show that a simple symmetric corrector improves the error of standard product formulas by two orders of magnitude. For perturbed systems, we show that correctors reduce the approximation error by a factor of the perturbation parameter.

%==================================================
\subsubsection{Higher-order CPFs for non-perturbed systems}
\label{subsubsec:higher_order_comp_corrector}
%==================================================

The CPF2 with a symmetric corrector developed in~\cref{subsec:CPF2} has time-reversal symmetry.
This symmetry allows CPF2 to serve as the base case to recursively construct a CPF of order $2k$ (CPF$2k$) from a CPF of order $2k-2$.
The approximation error of the resulting CPF$2k$ scales as $\order{|\lambda|^{2k+3}}$, which is two orders better than the $\order{|\lambda|^{2k+1}}$ error of the standard PF$2k$.
Here, we present the recursive construction of high-order CPFs applicable to non-perturbed systems and establish the improvement offered by correctors.

We begin with constructing CPF4 from CPF2.
Let $S^c_2(\lambda)$ denote the symmetric CPF2 in~\cref{eq:Csym_base}.
By the time-reversal symmetry, we have
\begin{equation}
    \log S^c_2(\lambda) = \lambda(A+B)+
    \sum_{j\geq2} \lambda^{2j+1} E_{2j+1},
\end{equation}
where the error operator $E_{2j+1}$ comprises some nested commutators of depth~$2j$.
As Suzuki's formula of the forth order~\cite{HS05}, we construct the CPF4 as
\begin{equation}
\label{eq:CPF4sym}
    S^c_{4}(\lambda) :=
    \left[S^c_2 (a_2 \lambda)\right]^2
    S^c_2 ( (1-4a_2) \lambda)
    \left[S^c_2 (a_2 \lambda)\right]^2
\end{equation}
for some appropriately chosen $a_2$.
By the Taylor expansion, we obtain
\begin{equation}
    \log S^c_4(\lambda) = \lambda(A+B) + \left[4a_2^5
    + \left(1-4a_2^5\right)
    \right]\lambda^5E_5 + \order{|\lambda|^7},
\end{equation}
so setting $4 a_2^5 + 4 (1-4a_2)^5=0$, or $a_2 = 1/(4-4^{1/5})$, yields the CPF4 with error $\order{|\lambda|^7}$.
Observe that this CPF4 retains the time-reversal symmetry; i.e., $S^c_4(\lambda)S^c_4(-\lambda)=\id$.
Hence, CPF6 can be constructed similarly from CPF4.
In general, the CPF$2k$ is recursively constructed~as
\begin{equation}
\label{eq:CPF2k_sym}
    S^c_{2k}(\lambda) :=
    \left[S^c_{2k-2} (a_k \lambda)\right]^2
    S^c_{2k-2} ( (1-4a_k) \lambda)
    \left[S^c_{2k-2} (a_k \lambda)\right]^2
    \qquad (\forall k\geq 2),
\end{equation}
and setting
\begin{equation}
    4 a_k^{2k+1} +
    4(1-4a_k)^{2k+1}=0, 
    \quad
    a_k = \frac{1}{4-4^{1/(2k+1})},
\end{equation}
asserts that the approximation error of CPF$2k$ scales as $\order{|\lambda|^{2k+3}}$.

%==================================================
\subsubsection{Higher-order CPFs for perturbed systems}
\label{subsubsec:CPFs_pert_sys}
%==================================================

While CPFs for non-perturbed systems provide a two-order improvement in the approximation error, here we show a variant of higher-order CPFs constructed from a CPF2 with a symplectic corrector is more advantageous for perturbed systems.
For such systems, we construct a CPF$2k$ with error scaling as $\order{\alpha^2|\lambda|^{2k+1}}$, where~$\alpha$ is the perturbation parameter.
Compared with the $\order{\alpha|\lambda|^{2k+1}}$ error of the standard PF$2k$, we see an improvement by a factor of $\alpha$ for any order~$2k$.
Notably, all intermediate correctors cancel out because of the symplectic property, so the overall simulation cost increases only by an additive constant.

To this end, we utilize two key observations: a CPF2 with a symplectic corrector has time-reversal symmetry, and the symplectic corrector removes the first order in $\alpha$ in all orders of $\lambda$ up to the $2k$th order.
To formalize this observation, we note that CPF2 with the symplectic corrector in \cref{eq:CPF2_pert_order2k} can be written as
\begin{equation}
    \log S^c_2(\lambda) =
    \lambda(A+\alpha B)
    +\sum_{j=1}^{k} \lambda^{2j+1} E'_{2j+1}
    +\order{\alpha |\lambda|^{2k+3}},
\end{equation}
where $E'_{2j+1}$ is the error operator of order $2j+1$ comprised of some nested commutators of depth $2j$ excluding the term $\ad^{2j}_A(\alpha B)$ that has the largest magnitude.
Using this CPF2 as the base case, the fourth-order CPF can be constructed by \cref{eq:CPF4sym} but with $a_2$ replaced with $p_2=1/(4-4^{1/3})$.
Because $\norm{E'_{2j+1}}\leq \alpha^2$, we have
\begin{equation}
    \log S^c_4(\lambda) =
    \lambda(A+\alpha B) 
    +\sum_{j=2}^{k} \lambda^{2j+1} E^{''}_{2j+1} +\order{\alpha|\lambda|^{2k+3}}
    = \lambda(A+\alpha B)
    +\order{\alpha^2|\lambda|^5}.
\end{equation}
Notice that CPF4 preserves the time-reversal symmetry; therefore, we can construct CPF6 from CPF4 in a similar way.
In general, CPF$2k$ has the time-reversal symmetry and is recursively constructed as
\begin{equation}
    S^c_{2k}(\lambda) :=
    \left[S^c_{2k-2} (p_k \lambda)\right]^2
    S^c_{2k-2} ( (1-4p_k) \lambda)
    \left[S^c_{2k-2} (p_k \lambda)\right]^2
\end{equation}
with $p_k$ given in \cref{eq:PF2k}.
The approximation error~$\order{\alpha^2|\lambda|^{2k+1}}$ is better than the error of the standard PF$2k$ by a factor of $\alpha$ for any $k$.

Note that the CPF4 and higher-order CPFs constructed by the above approach are built from a CPF2 with a symplectic corrector.
This approach reduces the error by a factor of~$\alpha$, but it also introduces a small multiplicative factor to the overall simulating cost.
In contrast, constructing CPFs with a symplectic corrector would only result in a negligible additive constant to the overall simulation cost.
In the following, we construct a CPF4 with a symplectic corrector.
First, we expand the kernel of standard PF4 in \cref{eq:PF4} using \cref{eq:symC} as
\begin{equation}
\label{eq:K4}
    K_4:= \log S_4(\lambda)
    =\log \left(\e^X \e^Y \e^{X}\right)
    = 2X+Y-\frac16[X+Y,X,Y] + \order{|\lambda|^7}
\end{equation}
with $X$ and $Y$ defined as
\begin{align}
    X &:= 2\log S_2(s\lambda)=
    2s\lambda H - \frac2{24} (s\lambda)^3
    [A+2B,A,B] + 2(s\lambda)^5 E_5 + \order{|\lambda|^7},\\
    Y &:= \log S_2((1-4s)\lambda)
    = (1-4s)\lambda H -\frac1{24} (1-4s)^3 \lambda^3 [A+2B,A,B] + (1-4s)^5 \lambda^5 E_5 + \order{|\lambda|^7},
\end{align}
where the right-hand sides follow from \cref{eq:PF2_kernel}.
We now show
\begin{equation}
\label{eq:K4exp}
    K_4 \equiv_{(\geq2)} \lambda H
    + c\lambda^5 \ad^4_A(B)
    +\order{|\lambda|^7},
\end{equation}
where $c$ is given in~\cref{eq:constants} and $\equiv_{(\geq2)}$ denotes equality modulo terms with degree $\geq 2$ in $B$.
To this end, note that~$s$ in~\cref{eq:constants} is chosen to cancel the third-order error term from the kernel of PF4.
Hence we have
\begin{align}
    2X+ Y &= \lambda H  +
    \left(4s^5 + (1-4s)^5\right)
    \lambda^5 E_5 + \order{|\lambda|^7}\\
    &\equiv_{(\geq2)}\lambda H  +
    \frac7{5760}\left(4s^5 + (1-4s)^5\right) \lambda^5\ad^4_A(B),
\end{align}
where the second line follows because by \cref{eq:K2}
\begin{equation}
    E_5 \equiv_{(\geq2)} \frac1{4!} B_4\left(\frac12\right) \ad^4_A (B)
    = \frac7{5760}\ad^4_A (B)
\end{equation}
with the value of $B_4(1/2)$ given in \cref{eq:bernoulli}.
Furthermore, observe that
\begin{align}
    [X,Y] &\equiv_{(\geq2)}
    \lambda^4
    \left[2s A,\,
    -\tfrac1{24}(1-4s)^3\ad^2_A(B)
    \right]
    + \lambda^4
    \left[
    -\tfrac2{24}s^3\ad^2_A(B),\,(1-4s)A
    \right] + \order{|\lambda|^6}\\
    & = -\frac1{12}
    \left[s(1-4s)^3 -s^3(1-4s)\right]
    \lambda^4\ad^3_A(B) + \order{|\lambda|^6}
\end{align}
and the third term for $K_4$ in \cref{eq:K4} is
\begin{align}
    -\frac16 [X+Y,X,Y]
    &= -\frac16 \left[(1-2s)\lambda H +\order{|\lambda|^3}, [X,Y]\right]\\
    &\equiv_{(\geq2)} \frac1{72} s (1-2s)(1-3s)(1-4s)(1-5s)\, \lambda^5 \ad^4_A(B) + \order{|\lambda|^7}.
\end{align}
Altogether, we obtain \cref{eq:K4exp} for $K_4$.
The symplectic corrector $C:=c\lambda^4\ad^3_A(B)$ modifies $K_4$ by \cref{eq:sympC} as
\begin{equation}
    K'_4 = K_4 + [C,K_4] + \order{|\lambda|^7}
    = K_4 + c\lambda^4 \left[\ad^3_A(B),\,
    \lambda H + \order{|\lambda|^5}\right]
    +\order{|\lambda|^7}
    \equiv_{(\geq2)}
    \lambda H 
    +\order{|\lambda|^7}.
\end{equation}
Proposition~\ref{prop:CPF4} follows from the above discussion and by replacing $B$ with $\alpha B$.

%==================================================
\section{Correctors for Yoshida-based product formulas}
\label{sec:yoshidaPFs}
%==================================================
In this section, we develop correctors for product formulas obtained based on Yoshida's method~\cite{yos90}.
Similar to the standard product formulas, higher-order product formulas in this method are obtained from the second-order formula but with a smaller number of exponentials.
Specifically, instead of using the recursive formula in \cref{eq:PF2k},
Yoshida~\cite{yos90} uses the ansatz
\begin{equation}
\label{eq:Sm}
    S^{(m)}(\lambda) =
    \left(\prod_{j=1}^m S_2(w_{m-j+1}\lambda)\right)
    S_2(w_0 \lambda)
    \left(\prod_{j=1}^m S_2(w_j\lambda)\right)
\end{equation}
to construct higher-order product formulas from the second-order formula $S_2(\lambda)$ given in \cref{eq:PF2}.
In particular, here the parameters $m\geq 0$ and $w_j\in\mathbb{R}$ for $j\in\{0,1,2,\ldots,m\}$, need to be determined so that $S^{(m)}(\lambda)$ yields an order-$k$  product formula. 
To this end, one needs to solve a set of simultaneous nonlinear polynomial equations.
The polynomial equations do not have a unique solution, resulting in several product formulas for a given order~$k$.
By this method, Yoshida~\cite{yos90} constructed 6th-order product formulas and only some of 8th-order product formulas.
Several works have since pushed the search to higher orders and found new solutions~\cite{KL97,Sofroniou2005integrators,BCM06,BCM08,MCP+24}.
In particular, Ref.~\cite{Sofroniou2005integrators} established several 10th-order formulas and Ref.~\cite{MCP+24} 
found highly accurate 8th-order formulas and also discovered new 10th-order formulas.
For convenience, hereafter we use YPF$k$ to denote the order-$k$ product formula generated by Yoshida's method.

We focus on constructing corrected YPFs that apply to perturbed systems with Hamiltonian $H=A+\alpha B$.
To this end, first we analyze the kernel of YPFs in \cref{subsec:kernel_YPFs} and derive an expression for the kernel modulo terms with degree~$\geq2$ in~$B$.
We then construct corrected YPFs by two approaches.
In the first approach, described in \cref{subsec:YPFs_symp_correctors}, we utilize the derived expression for the kernel to construct symplectic correctors that generate corrected YPFs.
In the second approach, covered in \cref{subsec:YPFs_from_CPF2}, we construct corrected YPFs from a corrected second-order product formula.

A summary of correctors developed by these approaches and their effect on YPFs is provided in \cref{tab:correctorsYPF}.

\begin{table}[htb]
\centering
  \begin{tabular}{llll}
    \toprule
    Product &Error bound for &\multirow{2}{*}{Correctors} &Error bound for\\
    formula& non-corrected YPF&&corrected YPF\\
    \midrule
    YPF6 &$\order{\alpha|\lambda|^7}$&
    $C_\text{symp} = c\lambda^6 \ad^5_A(\alpha B)$
    with $c$ in \cref{eq:Km6mod}.
    &$\order{\alpha^2|\lambda|^7+\alpha|\lambda|^9}$\\
    YPF8 &$\order{\alpha|\lambda|^9}$ &
    $C_\text{symp}= c\lambda^8 \ad^7_A(\alpha B)$ with $c$ in \cref{eq:Km8mod}.
    &$\order{\alpha^2|\lambda|^9 + \alpha|\lambda|^{11}}$
    \\
    \multirow{3}{*}{$\begin{aligned}
        &\text{YPF$2k$}\\
        &k=3,4,5
    \end{aligned}$}
    &\multirow{3}{*}{$\order{\alpha|\lambda|^{2k+1}}$}
    &$C_\text{symp}= \displaystyle \sum_{j=1}^{k} \frac{B_{2j}(1/2)}{(2j)!} (w_\ell \lambda)^{2j}\ad^{2j-1}_A(\alpha B)$
    &\multirow{3}{*}{$\order{\alpha^2|\lambda|^{2k+1}
    +\alpha|\lambda|^{2k+3}}$}
    \\
    && used in the base case $S^c_2(w_\ell\lambda)$ in \cref{eq:SmC}. &\\
    \bottomrule
  \end{tabular}
  \caption{\label{tab:correctorsYPF}
  Correctors and error bounds of (non-)corrected Yoshida-based product formulas~(YPFs) for perturbed systems with Hamiltonian $H=A+\alpha B$, where $\alpha\ll 1$ and where partitions $A$ and $B$ have comparable norms.}
\end{table}
%==================================================
\subsection{The kernel of YPFs}
\label{subsec:kernel_YPFs}
%==================================================

We begin by deriving an expression for the kernel of YPFs with terms that have degree~$\leq 1$ in~$B$.
The kernel of $S^{(m)}(\lambda)$ in \cref{eq:Sm} up to the 10th order follows by \cite[Eq.~(5.2)]{yos90} and \cite[Eq.~(A17)]{MCP+24} as
\begin{align}
\label{eq:kernel_YPFs}
    K^{(m)} := \log S^{(m)}(\lambda)
    &=\lambda A_{1,m} \alpha_1 + \lambda^3 A_{3,m}\alpha_3 + \lambda^5( A_{5,m}\alpha_5 + B_{5,m} \beta_5 )\nonumber\\
    &\quad +\lambda^7 (A_{7,m}\alpha_7 + B_{7,m} \beta_7 + C_{7,m}\gamma_7 + D_{7,m}\delta_7)\nonumber\\
    &\quad + \lambda^9 (A_{9,m}\alpha_9 + B_{9,m}\beta_{9} + C^{(1)}_{9,m}\gamma^{(1)}_{9} + C^{(2)}_{9,m} \gamma^{(2)}_9 +  C^{(3)}_{9,m} 
    \gamma^{(3)}_9\nonumber\\
    &\quad + D^{(1)}_{9,m} \delta^{(1)}_9 + D^{(2)}_{9,m} \delta^{(2)}_9 + D^{(3)}_{9,m} \delta^{(3)}_9 + E_{9,m} \epsilon_9) + 
    \order{|\lambda|^{11}},
\end{align}
where the variables in uppercase denote polynomials in the scalar variables $(w_1,\ldots,w_m)$
and the variables in Greek letters, except $\lambda$, denote some nested commutators, which are explicitly defined below.
For instance, the polynomials $A_{j,m}$ are defined as $A_{j,m} := w^j_0 + 2\sum_{\ell=1}^m w^j_\ell.$
We refer to Ref.~\cite[Eqs.~(5.8--5.11)]{yos90} for expressions of the rest of the polynomials used in the $\lambda^5$ and $\lambda^7$ terms and to Ref.~\cite[Eqs.~(A38--A45)]{MCP+24} for those used in the $\lambda^9$ term.
Below we state the expressions for the nested commutators $\alpha_j$, $\beta_j$, $\gamma_j$
and provide equivalent expressions for them
modulo terms with degree~$\geq2$ in the operator~$B$.
Then we use these expressions to construct the correctors.
As usual in this work, the symbol~$\equiv_{(\geq 2)}$ used in the rest of this section denotes equality modulo terms with degree $\geq 2$ in the operator $B$.

The nested commutators $\alpha_j$ are defined such that
\begin{equation}
\label{eq:alpha_j}
  \log \left(\e^{A/2}\e^B\e^{A/2}\right) =
  \sum_{\ell=0}^\infty \alpha_{2\ell+1},
\end{equation}
and for all $\alpha_j$ used in \cref{eq:kernel_YPFs} we have (see \cite[Eqs~(6--8) and Eq.~(A16)]{MCP+24})
\begin{align}
\label{eq:alpha1}
    \alpha_1 &= A+B,\\%%%
    \alpha_3 &= - \frac1{24}\ad^2_A(B)
    +\frac1{12}\ad^2_B(A)
    \equiv_{(\geq 2)}
    - \frac1{24}\ad^2_A(B)
    = \frac1{2!}B_2\left(\frac12\right)\ad^2_A(B),\\%%%
    \alpha_5 &= \frac7{5760}\ad^4_A(B)
    -\frac1{720}\ad^4_B(A)
    +\frac1{360}\ad_A\left(\ad^3_B(A)\right)
    +\frac1{360}\ad_B\left(\ad^3_A(B)\right)\nonumber\\
    &\quad -\frac1{480}\ad^2_A\left(\ad^2_B(A)\right)
    +\frac1{120}\ad^2_B\left(\ad^2_A(B)\right)\nonumber\\
    &\equiv_{(\geq 2)}
    \frac7{5760}\ad^4_A(B)
    = \frac1{4!}B_4\left(\frac12\right)\ad^2_A(B),\\%%%
    \alpha_7 &=
    -\frac{31}{967680}\ad^6_A(B)
    -\frac{31}{161280}\ad_B\left(\ad^5_A(B)\right)
    -\frac{13}{30240}\ad^2_B\left(\ad^4_A(B)\right)\nonumber\\
    &\quad -\frac{53}{120960}\ad^3_B\left(\ad^3_A(B)\right)
    -\frac{1}{5040}\ad^4_B\left(\ad^2_A(B)\right)
    -\frac{1}{30240}\ad^5_B\left(\ad_A(B)\right)\nonumber\\
    &\equiv_{(\geq 2)}
    -\frac{31}{967680}\ad^6_A(B)
    =\frac1{6!}B_6\left(\frac12\right)\ad^6_A(B).
\end{align}
As before, $B_{2j}(x)$ are Bernoulli polynomials in \cref{eq:bernoli}; few values at $x=1/2$ are given in \cref{eq:bernoulli}.
Note that the prefactor of the equivalent expressions modulo terms with degree~$\geq2$ match the prefactor of the corrector in~\cref{eq:CPF2symp}.
Similar to the above formulas, the expression for~$\alpha_9$ modulo terms with degree~$\geq2$ is
\begin{equation}
\label{eq:alpha9}
    \alpha_9 \equiv_{(\geq 2)}
    \frac1{8!} B_8\left(\frac12\right) 
    \ad^8_A(B).
\end{equation}
By the above formulas and~\cite[Eqs.~(A3--A14)]{MCP+24} we have
\begin{equation}
\label{eq:beta5}
    \beta_5=[\alpha_1,\alpha_1,\alpha_3] 
    \equiv_{(\geq 2)}
    -\frac{1}{24}\ad^4_A(B),
\end{equation}
for $\beta_5$ in the $\lambda^5$ term of the kernel in \cref{eq:kernel_YPFs};
\begin{align}
    \beta_7&=[\alpha_1,\alpha_1,\alpha_5]
    \equiv_{(\geq 2)}
    \frac{7}{5760}\ad^6_A(B),\\
    \delta_7&=[\alpha_1,\alpha_1,\alpha_1,\alpha_1,\alpha_3]
    \equiv_{(\geq 2)}-\frac{1}{24}\ad^6_A(B),\\
    \gamma_7&=[\alpha_3,\alpha_3,\alpha_1]
    \equiv_{(\geq 2)} 0,
\end{align}
used in the $\lambda^7$ term of the kernel;
and
\begin{align}
    \beta_9 &= [\alpha_1,\alpha_1,\alpha_7]
    \equiv_{(\geq 2)}
    -\frac{31}{967680}\ad^8_A(B), \\
    \gamma^{(1)}_9&=[\alpha_1,\alpha_3,\alpha_5]
    \equiv_{(\geq 2)} 0,\\
    \gamma^{(2)}_9&=[\alpha_3,\alpha_1,\alpha_5]
    \equiv_{(\geq 2)} 0,\\
    \gamma^{(3)}_9&=[\alpha_5,\alpha_1,\alpha_3]
    \equiv_{(\geq 2)} 0,\\
    \delta^{(1)}_9& = [\alpha_1^4, \alpha_5] =
    \ad^4_{\alpha_1}(\alpha_5)
    \equiv_{(\geq 2)}
    \frac{7}{5760}\ad^4_A(B),\\
    \delta^{(2)}_9&=[\alpha_3,\alpha_1^3,\alpha_3]
    \equiv_{(\geq 2)} 0, \\
    \delta^{(3)}_9&=[\alpha_1,\alpha_3,\alpha_1^2,\alpha_3]
    \equiv_{(\geq 2)}0,\\
    \epsilon_9 &\label{eq:eps9}
    = [\alpha_1^6,\alpha_3] = \ad^6_{\alpha_1}(\alpha_3)
    \equiv_{(\geq 2)}
    -\frac{1}{24}\ad^8_A(B),
\end{align}
used in the $\lambda^9$ term of the kernel.
The equivalent expressions for these nested commutators yield an expression for the kernel of YPFs in \cref{eq:kernel_YPFs} modulo terms with degree~$\geq2$ in~$B$.

%==================================================
\subsection{Corrected YPFs by symplectic correctors}
\label{subsec:YPFs_symp_correctors}
%==================================================

We now develop symplectic correctors for 6th-order~(YPF6) and 8th-order~(YPF8) product formulas based on Yoshida's method.  
We remark that several product formulas can be generated by Yoshida's method for a given order~$k$. Nonetheless, the corrector we develop for a given order applies to all product formulas in that order but with a constant factor specified by the particular product formula used.

We begin with the corrector for YPF6. Note that the leading error for the kernel of YPF6 is of seventh order in $\lambda$.
Specifically, by setting
\begin{equation}
\label{eq:sle5}
    A_{1,m} = 1,
    \quad
    A_{3,m} = 0,
    \quad
    A_{5,m}=0
    \quad
    \text{and}
    \quad
    B_{5,m}=0
\end{equation}
in \cref{eq:kernel_YPFs}, the kernel of YPF6 is
\begin{align}
\label{eq:Km6}
    K^{(m)}_6 &= \lambda (A+B) + \lambda^7(A_{7,m}\alpha_7 + B_{7,m} \beta_7 + C_{7,m}\gamma_7 + D_{7,m}\delta_7) +\order{|\lambda|^{9}}\\
    &\label{eq:Km6mod}
    \equiv_{(\geq 2)}
    \lambda (A+B) + \lambda^7
    \underbrace{\left(-\frac{31}{967680} A_{7,m} + \frac{7}{5760} B_{7,m} -\frac{1}{24} D_{7,m}\right)}_{:=c} \ad^6_A(B)
    +\order{|\lambda|^9},
\end{align}
where the equality modulo terms with degree $\geq 2$ in $B$ is obtained from the formulas in the previous subsections.
Numerical values for the scalar variables $(w_1,\ldots,w_m)$ that enter the polynomials $A_{7,m}$, $B_{7,m}$ and $D_{7,m}$ are obtained by solving the set of algebraic equations in~\cref{eq:sle5}.
Yoshida provides three solutions for $m=3$~\cite[Table~1]{yos90} that appear to be all solutions for the 6th order; Ref.~\cite{MCP+24} also performed an extensive search and did not find additional solutions.
The expressions for the polynomials~$A_{7,m}$, $B_{7,m}$ and~$D_{7,m}$ are known and are given in Ref.~\cite[Eqs.~(5.8)--(5.11)]{yos90}, from which we obtain the numerical values for these polynomials and the numerical value for the constant $c$ defined in \cref{eq:Km6mod}.

Let us now assume we are given a perturbed system with the Hamiltonian $H=A+\alpha B$, where $0<\alpha\ll1$ is the perturbation parameter; for such systems $B$ is replaced with $\alpha B$ in the kernel.
We take the symplectic corrector as $C = c\lambda^6 \ad^5_A(\alpha B)$, where $c$ is the constant in \cref{eq:Km6mod}.
As per \cref{eq:sympC}, this symplectic corrector modifies the kernel $K^{(m)}_6$ in \cref{eq:Km6} as
\begin{equation}
    K'^{(m)}_6 = K^{(m)}_6 + [C, K^{(m)}_6] + \cdots
    \equiv_{(\geq 2)}
    \lambda H +\order{\alpha^2 |\lambda|^7
             +\alpha |\lambda|^9},
\end{equation}
yielding an improvement in the leading error by a factor of~$\alpha$.
In contrast, the leading error of the non-corrected YPF6 for perturbed systems scales as $\order{\alpha|\lambda|^7}$. 

We take a similar approach to construct a symplectic corrector for Yoshida's 8th-order product formula. The kernel of this product formula by~\cref{eq:kernel_YPFs} is
\begin{align}
\label{eq:Km8}
    K^{(m)}_8 &= \lambda (A+B) + \lambda^9 (A_{9,m}\alpha_9 + B_{9,m}\beta_{9} + C^{(1)}_{9,m}\gamma^{(1)}_{9} + C^{(2)}_{9,m} \gamma^{(2)}_9 +  C^{(3)}_{9,m} 
    \gamma^{(3)}_9 \nonumber\\
    &\quad + D^{(1)}_{9,m} \delta^{(1)}_9 + D^{(2)}_{9,m} \delta^{(2)}_9 + D^{(3)}_{9,m} \delta^{(3)}_9 + E_{9,m} \epsilon_9) + \order{|\lambda|^{11}}\\
    &\label{eq:Km8mod}
    \equiv_{(\geq 2)}
    \lambda (A+B)\! + \lambda^9\!
    \underbrace{
    \left(\frac1{8!} B_8\left(\frac12\right)A_{9,m}
    - \frac{31}{967680} B_{9,m}
    -\frac{7}{5760} D^{(1)}_{9,m}
    -\frac{1}{24} E_{9,m}
    \right)}_{:=c}\! \ad^8_A(B)
    +\order{|\lambda|^{11}},
\end{align}
where the equality modulo terms with degree $\geq 2$ in $B$ is obtained as before.
Here $A_{9,m}$, $B_{9,m}$, $D^{(1)}_{9,m}$ and~$E_{9,m}$ are the 9th-order polynomials in the variables $(w_1,\ldots,w_m)$.
These variables are obtained by solving a set of algebraic equations for which many solutions exist.
Five solutions with $m=7$ are provided in Ref.~\cite[Table~2]{yos90} and Ref.~\cite{MCP+24} found many more new solutions; see Ref.~\cite[Tables~I--III]{MCP+24} for some solutions with~$m=7,8,10$.
The 9th-order polynomials are given in Ref.~\cite[Eqs.~(A38--A45)]{MCP+24}, from which we obtain the numerical values for these polynomials and the constant~$c$ defined in \cref{eq:Km8mod}.

Let us now take the symplectic corrector as $C = c\lambda^8 \ad^7_A(\alpha B)$.
Then by~\cref{eq:sympC} we obtain the expression
\begin{equation}
    K'^{(m)}_8 = K^{(m)}_8 + [C, K^{(m)}] + \cdots
             \equiv_{(\geq 2)} \lambda H
             +\order{\alpha^2 |\lambda|^9
             +\alpha|\lambda|^{11}}
\end{equation}
for the modified 8th-order kernel.
Observe that the leading error here is better than the leading error~$\order{\alpha|\lambda|^9}$ for non-corrected YPF8 by a factor of~$\alpha$.

The described approach is applicable to 6th- and 8th-order YPFs, as the error operators of orders greater than~$11$ in the kernel are unknown.
Next we describe an approach that can be used for higher-order YPFs.
%==================================================
\subsection{Corrected YPFs built from a CPF2 with a symplectic corrector}
\label{subsec:YPFs_from_CPF2}
%==================================================

An alternative approach to constructing a corrected version of YPF6, YPF8, and YPF10 is to use a corrected version of the base product formula that generates these higher-order product formulas.
In this approach, we simply modify the ansatz in \cref{eq:Sm} as
\begin{equation}
\label{eq:SmC}
    S^{(m)c}(\lambda) =
    \left(\prod_{\ell=1}^m S^c_2(w_{m-\ell+1}\lambda)\right)
    S^c_2(w_0 \lambda)
    \left(\prod_{\ell=1}^m S^c_2(w_\ell\lambda)\right),
\end{equation}
where $S^c_2(w_\ell\lambda)$ is the corrected second-order product formula~(CPF2) with the symplectic corrector in \cref{eq:CPF2symp}.
For clarity, we restate the CPF2 as 
\begin{equation}
\label{eq:CPF2wl}
    S^c_2(w_\ell\lambda):=\e^{C(k,w_\ell)} S_2(w_\ell\lambda) \e^{-C(k,w_\ell)}
    \quad
    \text{with}
    \quad
    C(k,w_\ell) := \alpha \sum_{j=1}^{k} \frac{B_{2j}(1/2)}{(2j)!} (w_\ell\lambda)^{2j}\,\ad^{2j-1}_A(B)
\end{equation}
to show the corrector is a function of the scalar variables $w_\ell$.
The parameter $k$ here is chosen based on the order of the YPF.
Specifically, we choose $k=3,4,5$ for YPF6, YPF8 and YPF10, respectively.
As per Proposition~\ref{prop:high-order-CPF2}, such a corrector removes the error terms of order $\alpha |\lambda|^{2j+1}$ for $j={1,2,\ldots,k}$ from $S_2(\lambda)$.
The remaining error terms are of order $\alpha^2 |\lambda|^{2j+1}$ for $j\le k$ and of order $\alpha |\lambda|^{2j+1}$ for $j>k$.
That is to say, the term with degree one in~$B$ is removed from~$\alpha_j$ in Eqs.~(\ref{eq:alpha1}--\ref{eq:alpha9}) by the corrector. Consequently, the term with degree one in~$B$ is also removed from the nested commutators in Eqs.~(\ref{eq:beta5}--\ref{eq:eps9}).

In other words, the CPF2 maps the nested commutator $\alpha_{2j+1}$ in \cref{eq:alpha_j} with $j\leq k$ to $\widetilde{\alpha}_{2j+1}$, where~$\widetilde{\alpha}_{2j+1}$ does not have terms with degree one in~$B$.
Likewise, the nested commutators $\{\beta, \delta, \gamma,\ldots\}$ in Eqs.~(\ref{eq:beta5}--\ref{eq:eps9}) are mapped to nested commutators $\{\widetilde{\beta}, \widetilde{\delta}, \widetilde{\gamma},\ldots\}$ that do not have terms with degree one in~$B$.
Consequently, the kernel $K^{(m)}$ in \cref{eq:kernel_YPFs} is mapped to $\widetilde{K}^{(m)}$ that has error terms with degree $\geq2$ in $B$; all error terms degree one in $B$ are removed by the corrector. 
For example, the kernel of the 6th-order YPF in \cref{eq:Km6} is mapped~as
\begin{align}
    K^{(m)}_6 \mapsto \widetilde{K}_6^{(m)}
    \lambda H
    +\lambda^7 (A_{7,m}\widetilde{\alpha}_7 + B_{7,m} \widetilde{\beta}_7 + C_{7,m}\widetilde{\gamma}_7 + D_{7,m}\widetilde{\delta}_7) + \order{\alpha |\lambda|^9},
\end{align}
making the leading error of the corrected YPF6 scale as~$\order{\alpha^2|\lambda|^7+\alpha|\lambda|^9}$.
It simply follows that the leading error of corrected YPF2$k$ for $k\ge3$ by this approach scales as $\order{\alpha^2|\lambda|^{2k+1}+\alpha|\lambda|^{2k+3}}$, which is better than the error $\order{\alpha|\lambda|^{2k+1}}$ of non-corrected~YPF$2k$ by a factor of $\alpha$. 

The symplectic corrector used to correct the base product formula in~\cref{eq:CPF2wl} depends on the variables~$w_\ell$.
Therefore, the correctors for adjacent CPF2s in the modified ansatz in \cref{eq:SmC} do not cancel out.
However, the positive and negative components of the correctors can be combined to reduce the number of exponentials due to the correctors.
Specifically, for two adjacent second-order formulas we have
\begin{equation}
    S_2(w_\ell\lambda)S_2(w_{\ell'}\lambda)
    \mapsto
    S^c_2(w_\ell\lambda)S^c_2(w_{\ell'}\lambda)
    = \e^{C(k,w_\ell)}
    S_2(w_\ell\lambda)
    \e^{C'(k,w_\ell,w_{\ell'})}
    S_2(w_{\ell'}\lambda)
    \e^{-C(k,w_{\ell'})},
\end{equation}
where
\begin{equation}
    C'(k,w_\ell,w_{\ell'}):= -C(k,w_\ell) + C(k,w_{\ell'}) =
    \alpha \sum_{j=1}^{k} \frac{B_{2j}(1/2)}{(2j)!} \left(w_\ell^{2j}+w_{\ell'}^{2j}\right)\lambda^{2j}\,\ad^{2j-1}_A(B)
\end{equation}
is the combination of the negative and positive components of the correctors.
By this combination, the number of exponentials due to the correctors in the ansatz in \cref{eq:SmC} would be $2m+2$.

%==================================================
\section{Compilation for correctors}
\label{sec:compile}
%==================================================
The correctors we established are linear combination of nested commutators of~$A$ and~$B$.
To use CPFs on a quantum computer, we must compile the exponential of each corrector into efficiently implementable operations.
Here, we provide a compilation as a sequence of products of exponentials of~$A$ and~$B$.
The compilation produces an error, but we keep the error at or below the order of the approximation error by the CPF.
Our compilations here can also be used to simulate exponentials of linear combination of nested commutators~\cite{CW13,CCH22}, enabling efficient synthesis of complicated unitaries using a limited set of native gates.

Our compilations rely on two formulas.
The key formula is provided in the following lemma, followed by a proof.
The second is \cref{eq:symC} for the symmetric corrector.

\begin{lemma}
\label{lem:keyformula}
    $\e^X \e^Y \e^{-X}=\exp(\e^{\ad_X}Y)$.
\end{lemma}
\begin{proof}
By the Taylor expansion, we have
    \begin{equation}
        \exp\left(\e^X Y \e^{-X}\right) =
        \sum_{j=0}^\infty
        \frac{1}{j!} \left(\e^X Y \e^{-X}\right)^j
        = \sum_{j=0}^\infty
        \frac{1}{j!}  \e^X \big( Y^j \big) \e^{-X}
        = \e^X \e^Y \e^{-X}.
    \end{equation}
The lemma then follows by the well-known identity $\e^X Y \e^{-X} = \e^{\ad_X}Y$~\cite[Equation~21]{LR01}.
\end{proof}

\subsection{Compiling the correctors for PF1 and  PF2}
\label{subsec:compile_PF12}
We begin with a compilation for the corrector $C = c_2\lambda^2 [A,B]
+c_3 \lambda^3 [B,A,B]$ and build upon it to compile other correctors.
An instance of this corrector is the symmetric corrector for CPF1 in~\cref{eq:PF1Csym}.
The error of this CPF is $\order{|\lambda|^4}$, so we provide a compilation with an error at most $\order{|\lambda|^4}$.

Let $X(a,b):= \e^{a\lambda A}\e^{b\lambda B}\e^{-a\lambda A}$ and let $\mathbb{B} := \e^{\ad_{b\lambda B}} A$, then we have
\begin{align}
   Y(a,b)&:=X(a,b) X(-a,-b)\\
   &=\e^{a\lambda A} \e^{-2a\lambda\mathbb{B}} \e^{a\lambda A}
   \qquad \text{[by Lemma~\ref{lem:keyformula}]}\\
   &=\exp\left({-2a\lambda\mathbb{B}+2a\lambda A}
   +\tfrac{1}{3}(a\lambda)^3\left[A-2\mathbb{B},A,\mathbb{B}\right] + \order{|\lambda|^5} \right)
   \qquad \text{[by \cref{eq:symC}]}\\
   &\label{eq:Yab}
   = \exp\left(2ab\lambda^2[A,B]+ ab^2\lambda^3 [B,A,B] + \lambda^4 E_4 +\order{|\lambda|^5}\right),
\end{align}
where we used $\mathbb{B} = A + b\lambda [B,A] + \tfrac12(b\lambda)^2[B,B,A]+ \tfrac16(b\lambda)^3\ad^3_{B}(A)+\order{|\lambda|^4}$, and where
\begin{equation}
\label{eq:E4}
    E_4:= \frac13 ba^3 \ad^3_A(B)
    - \frac13 ab^3 \ad^3_B(A)
\end{equation}
is the error operator of the fourth order.
By setting appropriate values for parameters~$a$ and~$b$, we achieve the corrector $C$ in the exponent of \cref{eq:Yab} with the desired compilation error.
This is achieved by setting~$2ab=c_2$ and $ab^2=c_3$, which give $a=c_2^2/4c_3$ and $b=2c_3/c_2$.
Notice that $\exp(-C)$ here is implemented by $a \mapsto -a$.

We now build on $Y(a,b)$ to compile $\exp(C)$ with $C=c_3\lambda^3[B,A,B]$ up to 5th-order error;
$C_\text{sym}$ in \cref{eq:CPF2comp} is an example of this corrector with $c_3=-1/48$.
Observe that by \cref{eq:Yab} and \cref{eq:E4}
\begin{equation}
    Y(a,b) Y(a,-b) = \e^{2ab^2\lambda^3[B,A,B] + \order{|\lambda|^5}},
\end{equation}
so any $a$ and $b$ satisfying $2ab^2=c_3$ yield a valid compilation. We simply take $b=1$ and $a=c_3/2$.
To compile $\exp(c_2\lambda^2[A,B])$ with 4th-order error, we use the product~\cite{CCH22}
\begin{equation}
    \label{eq:compilation_commutaotr}
    W(a) =
    \e^{\frac{\sqrt{5}-1}{2} a\lambda A}  
    \e^{\frac{\sqrt{5}-1}{2}x\lambda B}
    \e^{a\lambda A}
    \e^{-\frac{\sqrt{5}+1}{2}\lambda B}
    \e^{\frac{3-\sqrt{5}}{2} a\lambda A}
    \e^{\lambda B} =
    \e^{a\lambda^2 [A,B] + \order{|\lambda|^4}}
\end{equation}
with $6$~primitive exponentials.
By this compilation and \cref{eq:symC}, we obtain the compilation
\begin{equation}
    e^{c_1\lambda B/2}W(c_2)e^{c_1\lambda B/2} =
    \e^{c_1\lambda B + c_2\lambda^2 [A,B] + \order{|\lambda|^4}}
\end{equation}
for $C=c_1\lambda B+c_2\lambda^2[A,B]$, which uses 7~primitive exponentials after merging adjacent exponentials.

Finally, we show the product with 7~primitive exponentials
\begin{equation}
\label{eq:P}
    P(a) = X(a/2,a)\, \e^{-2a\lambda B}\, X(-a/2,a)
    \quad\text{with}\quad a = (4c)^{1/3},
\end{equation}
provides a compilation for the exponential of $C=c \lambda^3[A+2B,A,B]$ with 5th-order error. An instance of this corrector is $C_\text{sym}$ in \cref{eq:Csym_base} with $c=1/48$, which serves as the base case for the high-order CPFs in \cref{eq:CPF2k_sym}.
Notably, the compilation $P(a)$ is time-reversal symmetric, and thus all error terms will have odd powers.
This symmetry is a crucial aspect of this compilation, as it is required for recursively constructing high-order CPFs. 
To see $P(a)$ compiles $\exp(C)$, observe that $P(a)=S_2(a\lambda)S_2(-2a\lambda)S_2(a\lambda)$, and by \cref{eq:PF2_kernel}
\begin{equation}
    \log P(a) = \frac14a^3 \lambda^3[A+2B,A,B]
    +\sum_{j\geq2} \lambda^{2j+1}E'_{2j+1}, 
\end{equation}
so setting $a^3=4c$ yields the desired compilation.
We summarize the correctors in
\cref{tab:compilations} along with their compilations and compilation errors. \cref{fig:compilation_errors} shows the compilation error of various correctors for perturbed and non-perturbed systems.

\begin{figure}
    \centering
    \includegraphics[width=.98\linewidth]{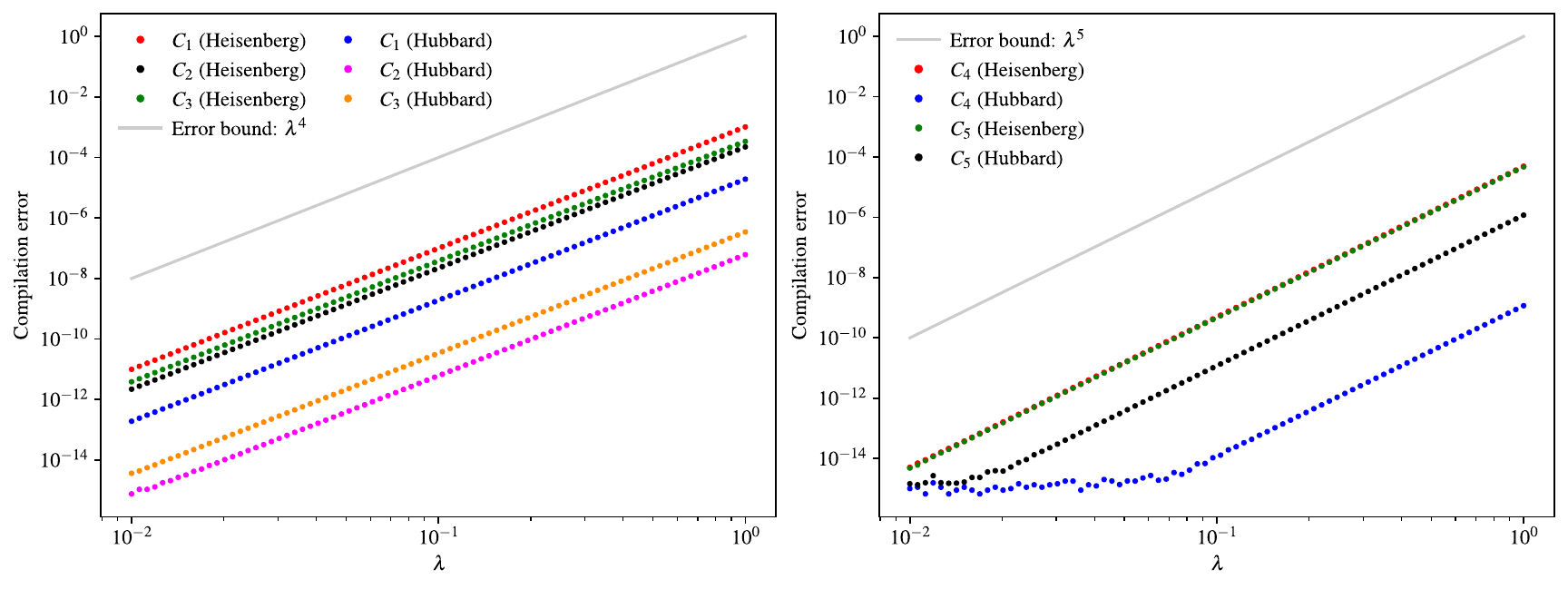}
    \caption{Compilation error of various correctors for perturbed (weak-coupling Hubbard model with $\alpha=0.1$) and non-perturbed (Heisenberg model) systems, with Hamiltonian terms normalized by the Hamiltonian norm. The correctors are
    $C_1=-\tfrac{x}4 -\tfrac{y}{12};
    C_2=-\tfrac{x}{24};
    C_3=\tfrac{\alpha}2\lambda B + \tfrac{x}{12}$;
    $C_4=\tfrac{y}{48}$ and
    $C_5=\tfrac{z}{48}-\tfrac{y}{24}$, where
    $x:=\lambda^2\ad_A(\alpha B)$,
    $y:= \lambda^3\ad^2_{\alpha B}(A)$ and
    $z:=\lambda^3\ad^2_{A}(\alpha B)$.
    Here we take $\lambda>0$ and numerically evaluate the error, the norm of the difference between~$\exp{(C)}$ and its compiled version.
    Solid lines indicate the error scaling.} 
    \label{fig:compilation_errors}
\end{figure}

\begin{table}[htb]
\renewcommand{\arraystretch}{1.6}
\centering
\begin{tabular}{lllcc}
\toprule
    Corrector
    &Compilation for $\exp(C)$
    &Compilation for $\exp(-C)$
    &Error
    &Cost\\
\midrule
     {\small$C = c_2\lambda^2\ad_A(B) + c_3\lambda^3\ad^2_B(A)$}
    &{\small$Y\left(-\tfrac{c^2_2}{4c_3},\tfrac{2c_2}{c_3}\right)$}
    &{\small$Y\left(\tfrac{c^2_2}{4c_3},-\tfrac{2c_2}{c_3}\right)$}
    &{\small$\order{|\lambda|^4}$}
    &{\small$5$}
    \\
    {\small$C = c_2\lambda^2\ad_A(B)$}
    &{\small$W(c_2)$}
    &{\small$W(-c_2)$}
    &{\small$\order{|\lambda|^4}$}
    &{\small$6$}
    \\
     {\small$C = c_3\lambda^3\ad^2_B(A)$}
    &{\small$Y\left(-\tfrac{c_3}{2},1\right)
    Y\left(-\tfrac{c_3}{2},-1\right)$}
    &{\small$Y\left(\tfrac{c_3}{2},1\right)
    Y\left(\tfrac{c_3}{2},-1\right)$}
    &{\small$\order{|\lambda|^5}$}
    &{\small$9$}
    \\
     {\small$C = c_1\lambda B + c_2\lambda^2 \ad_A(B)$}
    &{\small$\e^{c_1\lambda B/2} W(c_2)\, \e^{c_1 \lambda B/2}$}
    &{\small$\e^{-c_1\lambda B/2} W(-c_2)\, \e^{-c_1 \lambda B/2}$}
    &{\small $\order{|\lambda|^4}$}
    &{\small$7$}
    \\
    {\small$C = c\lambda^3 \ad^2_A(B) - 2c\lambda^3 \ad^2_B(A)$}
    &
    {\small$P(\sqrt[3]{4c})$}
    &
    {\small$P(-\sqrt[3]{4c})$}
    &
    {\small $\order{|\lambda|^5}$}
    &
    {\small$7$}
    \\
     {\small$C = \displaystyle\sum_{j=1}^{k} \frac{B_{2j}(1/2)}{(2j)!} \lambda^{2j}\,\ad^{2j-1}_A(\alpha B)$}
    &{\small$\displaystyle{\prod_\ell}
    Y(a_\ell, b_\ell) \prod_\ell
    Y(-a_\ell, -b_\ell)$}
    &{\small$\displaystyle{\prod_\ell}
    Y(a_\ell, -b_\ell) \prod_\ell
    Y(-a_\ell, b_\ell)$}
    &{\small$\order{\alpha^3|\lambda|^3}$}
    &{\small$10k$}
    \\
  \bottomrule
\end{tabular}
\caption{\label{tab:compilations}
Compilation for various correctors with their associated error and cost. The error is measured by the spectral norm of the difference between the ideal and compiled correctors, and the cost by the number of primitive exponentials used.
$Y(a,b)$ is defined in~\cref{eq:Yab}; $W(a)$ in~\cref{eq:compilation_commutaotr}; and $P(a)$ in \cref{eq:P}.
For the last corrector, $a_\ell=\ell+1$ and $b_\ell$ are given in \cref{tab:solution}.
}
\end{table}

%==================================================
\subsection{Compiling the correctors for higher-order PFs}
\label{subsec:compile_PF2k}
%==================================================
Here we present a procedure for compiling the corrector given in \cref{eq:CPF2_pert_order2k} for high-order product formulas, and then extend it to compile a corrector of the form $C=c\lambda^{2m} \ad^{2m-1}_A(B)$ for some integer $m$ and constant $c$.
An instance of this corrector is used for CPF4 in~\cref{tab:correctors} and for the Yoshida-based product formulas of 6th and 8th orders given in~\cref{tab:correctorsYPF}.
We begin with the following proposition and use it to compile these~correctors.

\begin{proposition}
\label{prop:high_oerdr_comp}
    Let $Y(a,b) := X(a,b)X(-a,-b)$ with $X(a,b):= \e^{a\lambda A} \e^{b\lambda B}
    \e^{-a\lambda A}$.
    Then, for any $m\geq 1$,
    \begin{equation}
    \label{eq:high_oerdr_comp}
    \prod_{\ell=m-1}^{0}
    Y(a_\ell, b_\ell)
    \prod_{\ell=0}^{m-1}
    Y(-a_\ell, -b_\ell)
    = \e^C
    \quad
    \text{with}
    \quad
    C \equiv_{(\geq3)} 2\lambda \sum_{\ell=0}^{m-1} b_\ell
    \left[
    \exp(\ad_{a_\ell \lambda A})
    -\exp(\ad_{-a_\ell \lambda A})
    \right] B,
\end{equation}
where $\equiv_{(\geq3)}$ denotes equality modulo terms with degree $\geq 3$ in $B$.
\end{proposition}
\begin{proof}
We have $X(a,b)= \exp\left(b\lambda\exp\left(\ad_{a\lambda A}\right)B\right)$ by \cref{lem:keyformula}.
Let $\mathbb{A}^\pm_\ell := \exp\left(\ad_{\pm a_\ell\lambda A}\right)B$ and $\mathbb{D}_\ell:=\mathbb{A}^+_\ell
-\mathbb{A}^-_\ell$.
Then $X(\pm a_\ell,\pm b_\ell) = \exp(\pm \lambda b_\ell\mathbb{A}^\pm_\ell)$ and
    \begin{align}
        Y(a_\ell,b_\ell) Y(-a_\ell,-b_\ell) &=
        X(a_\ell,b_\ell) X(-a_\ell,-b_\ell)^2 X(a_\ell,b_\ell)\\
        &=\e^{\lambda b_\ell\mathbb{A}^+_\ell}
        \e^{-2\lambda b_\ell\mathbb{A}^-_\ell}
        \e^{\lambda b_\ell\mathbb{A}^+_\ell}\\
        &=\e^{2\lambda b_\ell(\mathbb{A}^+_\ell-\mathbb{A}^-_\ell)+ \tfrac13(\lambda b_\ell)^3[\mathbb{A}^+_\ell-2\mathbb{A}^-_\ell,\mathbb{A}^+_\ell,\mathbb{A}^-_\ell]+\,\cdots}
        \qquad \text{[by \cref{eq:symC}]}\\
        &= \e^{2\lambda b_\ell\mathbb{D}_\ell+\,\cdots},
    \end{align}
where ``$\cdots$" contains terms with degree $\geq 3$ in $B$.
This equation with $\ell = 0$ yields the first term for the corrector $C$ in \cref{eq:high_oerdr_comp}.
For the second term, we multiply $Y(a_1,b_1)$ from left and $Y(-a_1,-b_1)$ from right as
\begin{align}
    Y(a_1,b_1)\,\e^{2\lambda b_0\mathbb{D}_0+\,\cdots}\,
    Y(-a_1,-b_1)
    &=\e^{\lambda b_1\mathbb{A}^+_1}
    \e^{-\lambda b_1\mathbb{A}^-_1}
    \e^{2\lambda b_0\mathbb{D}_0+\,\cdots}
    \e^{-\lambda b_1\mathbb{A}^-_1}
    \e^{\lambda b_1\mathbb{A}^+_1}\\
    &=\e^{\lambda b_1 \mathbb{A}^+_1}
    \e^{2\lambda b_0\mathbb{D}_0
    -2\lambda b_1\mathbb{A}^-_1+\,\cdots}
    \e^{\lambda b_1\mathbb{A}^+_1}
    \qquad \text{[by \cref{eq:symC}]}\\
    &=\e^{2\lambda (b_0\mathbb{D}_0+b_1\mathbb{D}_1)+\,\cdots},
    \quad\qquad\qquad\qquad\;\;
    \text{[by \cref{eq:symC}]}
\end{align}
where ``$\cdots$" contains terms with degree $\geq 3$ in $B$, as before.
We can progressively add more terms to the corrector by repeating this process. 
For $m$ repetitions, we obtain
\begin{equation}
    \prod_{\ell=m-1}^{0}
    Y(a_\ell, b_\ell)
    \prod_{\ell=0}^{m-1}
    Y(-a_\ell, -b_\ell)
    = \e^{2\lambda \sum_\ell b_\ell \mathbb{D}_\ell +\,\cdots}.
\end{equation}
Equation~\eqref{eq:high_oerdr_comp} then follows by $\mathbb{D}_\ell =\left[\exp(\ad_{a_\ell \lambda A})
    -\exp(\ad_{-a_\ell \lambda A})\right] B$.
\end{proof}

To identify the set of compilation parameters $\{a_\ell\}$ and $\{b_\ell\}$, we now express the corrector $C$ in \cref{eq:high_oerdr_comp} in a form similar to that in \cref{eq:CPF2_pert_order2k}.
Let us expand the corrector $C$ in \cref{eq:high_oerdr_comp} as
\begin{align}
    C &\equiv_{(\geq3)} 2\lambda \sum_{\ell=0}^{m-1}
    b_\ell
    \left[
    \exp(\ad_{a_\ell \lambda A})
    -\exp(\ad_{-a_\ell \lambda A})
    \right] B\\
    &=2\lambda \sum_{j=0}^{\infty} \frac{\lambda^j}{j!}
    \left(\sum_{\ell=0}^{m-1}
    b_\ell \left[a^j_\ell - (-a_\ell)^j\right]\right)
    \ad^j_A(B)\\
    \label{eq:linsys}
    &= \sum_{j=1}^{\infty}
    \frac{\lambda^{2j}}{(2j-1)!}
    \left(\sum_{\ell=0}^{m-1}
    4b_\ell a^{2j-1}_\ell\right)
    \ad^{2j-1}_A(B).
\end{align}
Comparing with the high-order corrector in \cref{eq:CPF2_pert_order2k}, we obtain the set of linear equations
\begin{equation}
    \sum_{\ell=0}^{m-1}
    b_\ell a^{2j-1}_\ell = 
    \frac{B_{2j}\left(\tfrac12\right)}{8j} \quad \forall\, 1\leq j\leq k
\end{equation}
for $b_\ell$ given a set of values for $a_\ell$.
This set of equations can be expressed as a matrix equation
$A\vec{b}=\vec{B}$ as
\begin{equation}
\label{eq:mtxeq}
\begin{bmatrix}
    a_0& a_1& \cdots& a_{m-1} \\
    a^3_0& a^3_1& \cdots& a^3_{m-1}\\
    \vdots& \vdots& \ddots& \vdots\\
    a^{2k-1}_0& a^{2k-1}_1& \cdots& a^{2k-1}_{m-1}
\end{bmatrix}
\begin{bmatrix}
    b_0\\
    b_1\\
    \vdots\\
    b_{m-1}
\end{bmatrix}
= \frac18
\begin{bmatrix}
    B_2\left(\tfrac12\right)\\
    \tfrac12 B_4\left(\tfrac12\right)\\
    \vdots\\
    \tfrac1k B_{2k}\left(\tfrac12\right)
\end{bmatrix}.
\end{equation}
This matrix equation has a unique solution for $m=k$, the case with a square matrix, and the solution is nonzero for any set of nonzero values for~$a_j$ such that $a_j\neq a_{j'}$ for $j \neq j'$.
Hence, we take $m=k$, and for simplicity, we select $a_j=j+1$.
We decompose the matrix $A$ as $A = V(a^2_0, a^2_1,\ldots, a^2_{k-1}) D$, where $V$ is a Vandermonde matrix, defined as
\begin{equation}
    V(x_0,x_1,\ldots,x_{k-1}) :=
    \begin{bmatrix}
        1& 1& \cdots&1\\
        x_0& x_1& \cdots& x_{k-1}\\
        x^2_0& x^2_1& \cdots& x^2_{k-1}\\
        x^{k-1}_0& x^{k-1}_1& \cdots& x^{k-1}_{k-1}
    \end{bmatrix},
\end{equation}
and $D$ is the diagonal matrix $D:=\text{diag}\left(a_0,a_1,\ldots,a_{k-1}\right)$.
The solution of $A\vec{b}=\vec{B}$ is then $\vec{b} = D^{-1}V^{-1}\vec{B}$.
The Vandermonde matrix has an explicit inverse, and its entries are rational numbers for a set of rational numbers~$a_\ell$.
Specifically, the inverse of the Vandermonde matrix $V(a^2_0, a^2_1,\ldots, a^2_{k-1})$ can be written as~\cite{BP70}
\begin{equation}
    V^{-1} = \left(\prod_{j=0}^{k-2} L_j(1)^\top D_j\right)\left(
    \prod_{j=k-2}^0 L_j(a^2_j)\right),
\end{equation}
where the lower triangular matrix $L_j(x)$ and the diagonal matrix $D_j$ are defined as
\begin{equation}
    L_j(x) :=
    \begin{bmatrix}
        \id_j & \\
                &1 &\\
                &x &\ddots\\
                &  &\ddots & 1\\
                &  & &x
    \end{bmatrix},
    \quad
    D_j := \begin{bmatrix}
        \id_{j+1} \\
        & \frac1{a^2_{j+1}-
        a^2_0}\\
        && \frac1{a^2_{j+2}-
        a^2_1}\\
        &&&\ddots\\
        &&&&\frac1{a^2_{k-1}-
        a^2_{k-j-2}}
    \end{bmatrix}
\end{equation}
and $\id_j$ is the $j\times j$ identity matrix.
The solution $\vec{b}$ has rational entries for the chosen values $a_j=j+1$. \cref{tab:solution} provides the solution for a few values of $k$.

\begin{table}[htb]
\renewcommand{\arraystretch}{1.5}
\centering
\begin{tabular}{ lllll} 
 \toprule
 $k=1$ & $k=2$ & $k=3$ & $k=4$ & $k=5$ \\ 
 \midrule
 $b_0 =\frac{-1}{96}$&
 $b_0=\frac{-167}{11520}$&
 $b_0=\frac{-64457}{3870720}$&
 $b_0=\frac{-16705243}{928972800}$&
 $b_0=\frac{-1543769039}{81749606400}$\\ 
 &
 $b_1=\frac{47}{23040}$&
 $b_1 = \frac{3643}{967680}$&
 $b_1 = \frac{4732843}{928972800}$&
 $b_1 = \frac{10431823}{1703116800}$
 \\ 
 &&
 $b_2 = \frac{-1669}{3870720}$&
 $b_2 = \frac{-103343}{103219200}$&
 $b_2 = \frac{-28718033}{18166579200}$
 \\
 &&&
 $b_3=\frac{176509}{1857945600}$&
 $b_3=\frac{8177231}{30656102400}$\\
 &&&& $b_4=\frac{-2105933}{98099527680}$\\
 \bottomrule
\end{tabular}
\caption{\label{tab:solution}
The solution for the linear system in \cref{eq:mtxeq} with $m=k$ and $a_j=j+1$ for $1\le k \le 5$.}
\end{table}

We now use a modified version of the above approach to compile a corrector of the form~$c \lambda^{2m} \ad^{2m-1}_{A}(B)$,
where $c$ is a real number and $m$ is an integer; see \cref{tab:correctors} and \cref{tab:correctorsYPF} for instances of this corrector.
To this end, we need a set of numbers $a_\ell$ and $b_\ell$ that, by \cref{eq:linsys}, satisfy the set of equations
\begin{equation}
    \sum_{\ell=0}^{m-1}
    b_\ell a^{2j-1}_\ell =
    \begin{cases}
         0 \qquad\qquad\qquad  1\leq j< k\\
        \frac{c}4 (2m-1)! \qquad j=m.
    \end{cases}
\end{equation}
This set of equations is similar to the linear system in \cref{eq:mtxeq} but with $\vec{B}=\left(0,\ldots,0,\tfrac{c}4 (2m-1)!\right)$ as the right-hand-side vector.
The solution to this linear system yields a compilation for the corrector as in \cref{eq:high_oerdr_comp}.

%==================================================
\section{Numerical simulations}
\label{sec:numerics}
%==================================================
To demonstrate the efficacy of CPFs and validate the theoretical results in the previous sections, we perform numerical simulations for several perturbed and non-perturbed systems.
We compare the empirical performance of the order-$k$ standard product formula~(PF$k$) with that of its corrected version~(CPF$k$) for $k\in\{1,2,4,6\}$ by numerically computing the total simulation error, as quantified by the spectral norm of the difference between the exact time-evolution $\exp(-\i Ht)$ and its (corrected) product-formula approximation.

We consider two settings for the simulations, one with a fixed time step~$\tau$, and the other with a fixed number of steps~$r$.
The purpose of fixed-$r$ simulations is to illustrate the scaling of the error with the simulation time~$t$ as~$\tau$~varies.
The results for fixed-$\tau$ and fixed-$r$ simulations are shown in \cref{fig:fixed_tau_simulations} and \cref{fig:fixed_r_simulations}, respectively, with details of other parameters provided in the figure captions.
These figures present results for first- and second-order product formulas applied to one-dimensional~(1D) perturbed and non-perturbed systems.
In \cref{app:extended_simulations}, we cover the results for the 2D versions of these systems, higher-order formulas, and additional simulations over a broader parameter range to further demonstrate the generality of the improvements achieved by~CPFs.
Code for reproducing our simulations is available~in~\cite{codes}.

An important note for the fixed-$r$ simulations in \cref{fig:fixed_r_simulations} is that, with $r$ fixed, increasing $t$ moves the simulation toward a larger time step $\tau=t/r$, where the errors of different product formulas (and their corrected versions) naturally become closer.
To see this, consider the error estimates $E_1 \propto r\tau^2$ and $E_2\propto r\tau^3$ of standard PF1 and PF2, respectively, after~$r$ steps. Then $E_2/E_1\propto \tau$, so when~$\tau$ gets larger and closer to~1, the errors become less separated. In practice, if the goal were to simulate for longer times, one would increase~$r$.
The empirical errors also exhibit similar behavior at longer times.
For the same reason, the errors of corrected and standard product formulas can also get closer at longer times.
The fixed-$\tau$ simulations in \cref{fig:fixed_tau_simulations} show a better separation of errors at longer times, especially for perturbed systems.

We also note that, as seen in \cref{fig:fixed_r_simulations}, the empirical error of standard product formulas at longer times is often substantially smaller than their theoretical error estimate (obtained by multiplying the single-step error by the number of steps; see errors with $r\delta$ in the plots). This is because the basis error cancels between steps, whereas the eigenvalue error accumulates; see~\cite{MCP+24} for details of these errors.
Importantly, observe that the empirical error of CPFs is still smaller than the empirical error of the corresponding PFs at all times, especially for perturbed systems [see~\cref{fig:fixed_r_simulations}(c) and~\cref{fig:fixed_r_simulations}(f)]. The gap between errors of corrected and standard product formulas can become smaller when the standard product formulas exhibit good empirical error at longer times, especially for perturbed systems (\cref{fig:fixed_r_simulations}),
even though corrected product formulas for such systems have a provably improved error scaling by a factor of $\alpha$ for any time (see~\cref{tab:correctors}).

%==================================================
\subsection{Non-perturbed systems}
\label{subsec:nonperturbed}
%==================================================
We use the following lattice models as non-perturbed systems in our simulations.
For each model, we numerically evaluate the total simulation error for~$t$ that varies in the range $1\leq t\leq 10$. For each~$t$ in the fixed-$r$ simulations (\cref{fig:fixed_r_simulations}), we divide the evolution time~$t$ into $r=100$ steps, so the timestep $\tau:=t/r$ varies in the range $0.01\leq \tau\leq 0.1$. In the fixed-$\tau$ simulations (\cref{fig:fixed_tau_simulations}), we fix $\tau=0.1$ and increase $r$ from~$1$ to~$100$, so $t=r\tau$ ranges over $0.1\leq t\leq 10$.

\textbf{Heisenberg model.}
The first non-perturbed system we use is the Heisenberg model on a 1D lattice with $n$ sites and periodic boundaries.
This model's Hamiltonian is $H = \sum_{j=0}^{n-1} \Vec{\sigma}_j\cdot \Vec{\sigma}_{j+1}$,
where $\Vec{\sigma}_j=(X_j, Y_j, Z_j)$ is the vector of Pauli operators acting on site~$j$, and $\Vec{\sigma}_n=\Vec{\sigma}_0$ by periodicity.
We partition~$H$ as $H=A+B$,
where $A$/$B$ is the sum of Hamiltonian terms with even/odd~$j$. Note that terms within $A$ and $B$ mutually~commute.

\textbf{Ising model.}
The second non-perturbed system in our simulations is the transverse-field Ising model. The Hamiltonian of this model in 1D is
\begin{equation}
    \label{eq:ising}
    H = J H_{xx} + h H_z;
    \quad H_{xx} :=
    \sum_{j=0}^{n-2} X_j X_{j+1}
    + Y_0Z_1Z_2\cdots Z_{n-2}Y_{n-1};
    \quad H_z:= \sum_{j=0}^{n-1} Z_j,
\end{equation}
where~$J$ is the strength of the nearest-neighbor interaction, and~$h$ is the strength of the external field.
This Hamiltonian is an instance of the XY~model and can be analytically diagonalized~\cite[Eq.~5]{VCL09}.
The second term in $H_{xx}$ is a boundary term (for $n>2$) to impose conventional periodic boundary conditions~(i.e., $\hat{c}_{n}=\hat{c}_0$) on the fermionic operators $\{\hat{c}_j\}$ when the spin model is mapped to the fermionic representation~\cite[p.~6]{FC25}.
We adopt this choice to impose periodic boundaries to directly apply the formalism of~\cite{VCL09} for the exact time evolution of the Ising model.
Although this model is exactly solvable via efficient diagonalization, it serves as a good testbed for demonstrating the effect of correctors, especially on quantum hardware~(\cref{sec:hardware}).
We fix $J=h=1$ for this model in our numerical simulations for non-perturbed systems. This choice sets the unit of time for simulation.

\textbf{Hubbard model with intermediate coupling.}
The Hubbard model is an idealized Hamiltonian that captures qualitative aspects of high-temperature superconductors.
Its Hamiltonian in second quantization is
\begin{equation}
    H = - t_\text{hop}
    \sum_{\langle ij\rangle,\sigma}
    (c^\dagger_{i,\sigma} c_{j,\sigma} +
    c^\dagger_{j,\sigma} c_{i,\sigma}) +
    U_\text{int} \sum_j n_{j,\uparrow} n_{j,\downarrow},
\end{equation}
where $\sigma\in\{{\uparrow, \downarrow\}}$ labels the spin of fermions; $c_{j,\sigma}$ and $c^\dagger_{j,\sigma}$ are the annihilation and creation operators of the fermion with spin $\sigma$ on site $j$; $n_{j,\sigma}$ is the associated number operator; and $\langle ij\rangle$ denotes the sum over nearest-neighbor pairs~$i$ and~$j$.
The first sum (the kinetic part of $H$) describes the hopping of particles between neighboring sites, and the second sum (the potential part) describes on-site interaction with strength~$U_\text{int}$.
We take $t_\text{hop}, U_\text{int}>0$ in our numerical simulations, as these parameters typically have positive values for fermionic systems.

Depending on the values of $U_\text{int}$ and $t_\text{hop}$, the Hubbard Hamiltonian becomes a (non-)perturbed system.
In the intermediate-coupling regime, where  $U_\text{int}\approx t_\text{hop}$, the model is a non-perturbed system~\cite{QSA+22}.
In our numerical simulations, we consider the spinless model and set $U_\text{int}=t_\text{hop}$.
We write $H=A+B$ in this regime, and take $A$ as the interaction part and $B$ as the kinetic part of the Hamiltonian.

\begin{figure}
    \centering
    \includegraphics[width=.958\linewidth]{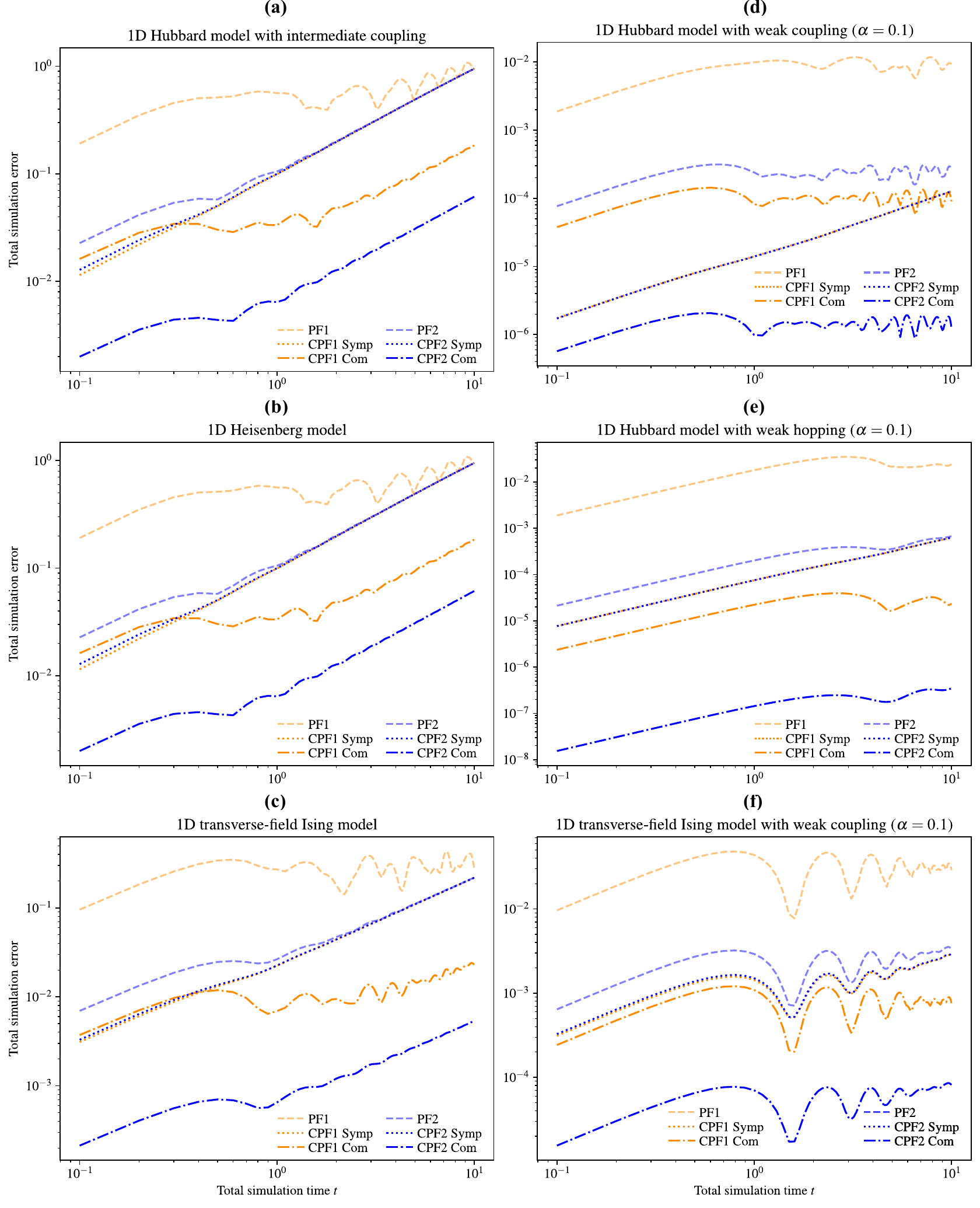}
    \caption{Empirical errors of PFs~(dashed lines) and CPFs (dotted and dash-dotted lines) for non-perturbed~[left:~(a)--(c)] and perturbed~[right:~(d)--(f)] systems in simulations with a fixed time step.
    Each system is on a 1D lattice of size~$n=8$ with~periodic boundaries.
    CPF1(2) Symp denotes the CPF1(2) with the symplectic corrector in the second (first) row of PF1(2) correctors in~\cref{tab:correctors}.
    CPF$x$~Com with $x\in\{1,2\}$ denotes CPF$x$ with the composite corrector in~\cref{tab:correctors}.
    Errors are evaluated at fixed $\tau=0.1$ for the number of steps $1\leq r\leq100$, so that $t=r\tau$ ranges over $0.1<t<10$.
    As expected from errors in \cref{tab:correctors}, PF2, CPF1~Symp, and CPF2~Symp perform similarly for non-perturbed systems~($\alpha=1$). Observe that CPFs are more advantageous for perturbed systems~($\alpha<1$).
    }
    \label{fig:fixed_tau_simulations}
\end{figure}

\begin{figure}
    \centering
    \includegraphics[width=.958\linewidth]{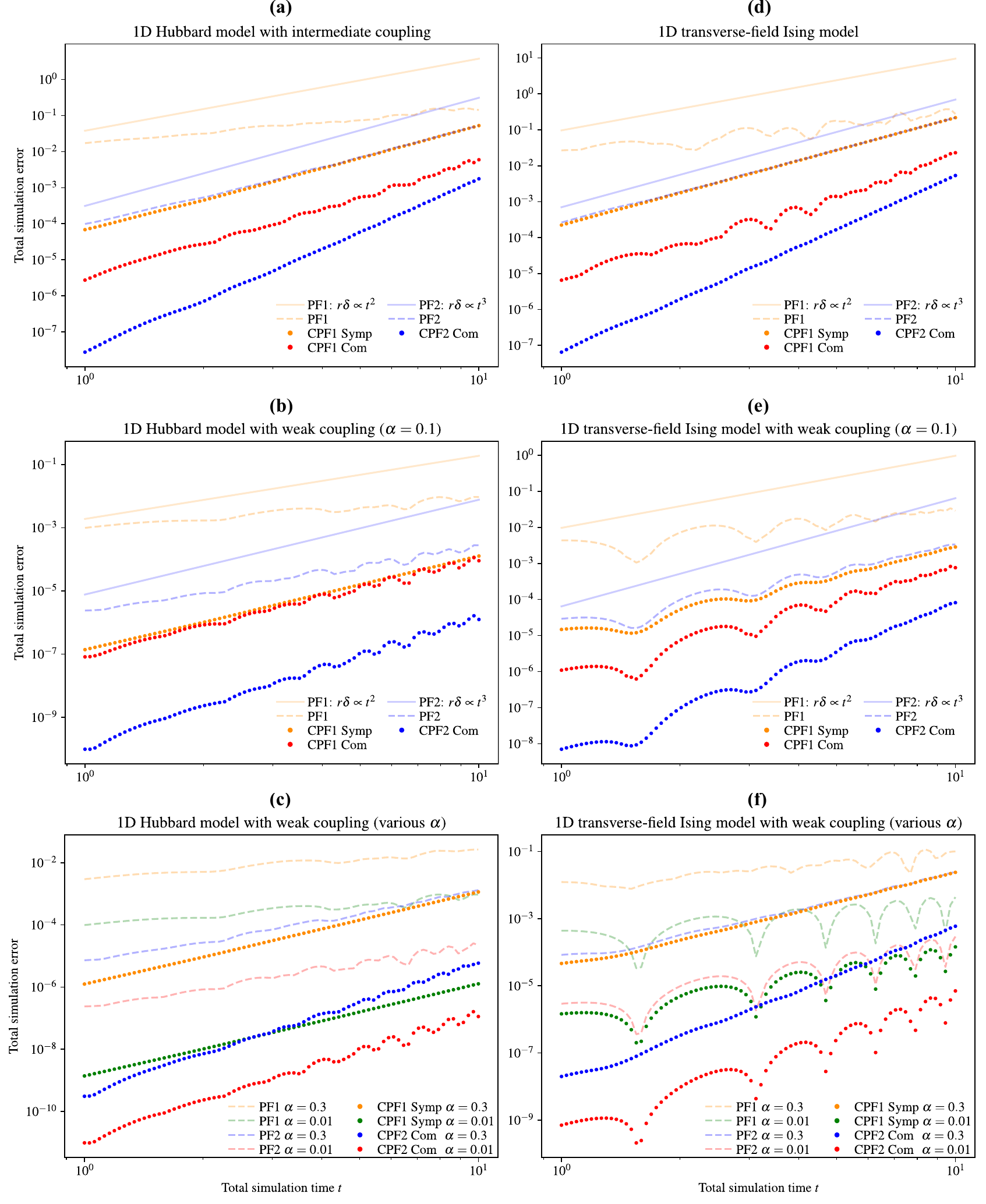}
    \caption{Empirical errors of PFs (dashed lines) and CPFs (dotted lines) for non-perturbed systems~[top:~(a),~(d)], and perturbed systems with fixed [middle:~(b),~(e)] and various~[bottom:~(c),~(f)] values of the perturbation parameter~$\alpha$.
    Errors are numerically evaluated for $1\leq t\leq 10$ using fixed~$r=100$ simulation steps; $\delta$ is the empirical error in one step, and the plots of~$r\delta$ (solid lines) show the theoretical scaling of the total error obtained from the triangle inequality.
    Each system is on a 1D lattice of size $n=8$ with periodic boundaries.
    Correctors used are the same as those in \cref{fig:fixed_tau_simulations}.
    As expected from errors in~\cref{tab:correctors}, CPF2~Symp performs similarly to CPF1~Symp and is therefore omitted.
    }
    \label{fig:fixed_r_simulations}
\end{figure}

%==================================================
\subsection{Perturbed systems}
\label{subsec:perturbed}
%==================================================
We use the following models as perturbed systems in our numerical simulations.
For each model, we fix the perturbation parameter to $\alpha=0.1$ and numerically compare the performance of PFs vs CPFs for various correctors.
We also consider cases with $\alpha=0.3$ and $\alpha=0.01$ in \cref{fig:fixed_r_simulations} to probe a broader range for the perturbation parameter.
All other simulation parameters are identical to our simulations for non-perturbed systems.
We note that small values of $\alpha$, such as $\alpha=0.01$, are more relevant for classical or high-precision simulations than for implementations on noisy hardware,
because the improvements provided by correctors for such cases can be well below current hardware noise.
Similarly, higher values of~$r$ (e.g., $r=1000$ used for simulations in \cref{app:extended_simulations}) are more appropriate for classical simulations or fault-tolerant quantum computers rather than near-term quantum computers.

\textbf{Weak Hubbard models.}
In regimes where the Hubbard model is a perturbed system, the Hamiltonian can be decomposed as $H=A+\alpha B$, where $A$ is the main part and $B$ is the perturbation part.
The ratio~$U_\text{int}/t_\text{hop}$ in the Hubbard model determines this decomposition.
We consider two cases:
\begin{itemize}
    \item\textbf{Weak coupling regime.}
    In this case, $U_\text{int} \ll t_\text{hop}$ and the kinetic part dominates; $A$ is the kinetic part and $B$ is the interaction term.
    In our numerical simulations, we fix $t_\text{hop}=1$ and set the perturbation parameter as $\alpha=U_\text{int}$.
    
    \item\textbf{Weak hopping (strongly correlated) regime}:
    In this case, $t_\text{hop}\ll U_\text{int}$ (often the regime of interest) and the potential term dominates; $A$ is the interaction part and~$B$ is the kinetic part.
    In our numerical simulation for this case, we fix $U_\text{int}=1$ and set the perturbation parameter as $\alpha= t_\text{hop}$.
\end{itemize}

\textbf{Ising model with weak coupling.}
The third perturbed system we use is the transverse-field Ising model in \cref{eq:ising} in the weak-coupling regime, where the interaction strength~$J$ is much smaller than the transverse field~$h$.
We set~$h=1$ in our simulations and treat~$J$ as the perturbation parameter ($J=\alpha \ll 1$), so time is measured in units of $1/h$; our simulations correspond to the dimensionless ratio $J/h=\alpha$.

%==================================================
\section{Quantum hardware implementations}
\label{sec:hardware}
%==================================================
In this section, we demonstrate the improvements offered by CPFs over the standard product formulas by implementations on actual quantum hardware as well as on noisy and noiseless quantum hardware simulators.
To this end, we compare the performance of CPF1 and CPF2 with symplectic correctors against the standard PF1 and PF2 for simulating the transverse-field Ising model. Due to hardware limitations, we restrict our implementations to small system sizes (i.e., small circuit width) and limit the circuit depth to the hardware's maximum reliable range.
We emphasize that the purpose of our hardware implementations is not to demonstrate quantum advantage, but rather to demonstrate and validate that corrected product formulas yield improved performance over standard product formulas even on current noisy hardware. The consistency between our theoretical analysis, classical simulations, and hardware experiments provides compelling evidence for practical value of corrected product formulas.

For implementations on quantum hardware, we use IBM's 127-qubit QPU \texttt{ibm\_quebec}\footnote{
This hardware has a median gate error rate of $7.258\times 10^{-3}$ for ECR (two-qubit) gates and $2.004\times 10^{-4}$ for SX (single-qubit) gates. The median T1 and T2 of this device is $311.83~\mu \mathrm{s}$ and $231.83~\mu\mathrm{s}$, respectively.
}, which is of the \emph{Eagle r3} processor family.
We use \texttt{Qiskit Aer} for noiseless hardware simulations and \texttt{FakeQuebec} for noisy  simulations, which provides a simulated version of the \texttt{ibm\_quebec} QPU.
We use these hardware and simulators to run multiple quantum circuits of CPFs and standard PFs for different system sizes and produce similar plots to those in~\cref{fig:fixed_r_simulations}.

A key difference in the quantum hardware implementations compared to classical simulations is the metric we use to quantify the approximation error.
Motivated by hardware limitations, we use average infidelity as the error metric for hardware implementations. This metric is easier to compute than the spectral norm used in our classical simulations, which requires a full state tomography.
Computing the average infidelity requires exact implementation of the time-evolution operator.
We give quantum circuits for the exact time evolution of the Ising model in \cref{subsec:exact_ising} and describe our hardware implementations in~\cref{subsec:hardware_results}. Figure~\ref{fig:hardware} shows the results of our hardware implementations.

%==================================================
\subsection{Exact quantum circuit for Ising model}
\label{subsec:exact_ising}
%==================================================
As an instance of the XY~model, the transverse-field Ising model can be analytically diagonalized, and the quantum circuit for its exact evolution follows from the diagonalization steps~\cite{VCL09}. For the model with periodic boundaries, these steps are as follows. First, the Jordan-Wigner transformation is used to represent the $n$-site Ising Hamiltonian (a spin or qubit Hamiltonian) as a fermionic Hamiltonian with~$n$ interacting fermions.
The interacting fermionic Hamiltonian is then mapped to a free fermionic Hamiltonian by a fermionic Fourier transform followed by a Bogoliubov transform.
The free fermionic Hamiltonian is a diagonal Hamiltonian that can be expressed as the qubit Hamiltonian
\begin{equation}
    D = \sum_{k=0}^{n-1}
    \omega_k Z_k,
    \quad w_k:= \sqrt{(h-J\cos(2\pi k/n))^2 + J^2\sin^2(2\pi k/n)}.
\end{equation}
The diagonal evolution under $D$ for time $\tau$ can be implemented by applying a $z$-rotation gate $R_z (\phi_k)=\exp(-\i \phi_kZ/2)$ on qubit $k$ with angle $\phi_k = 2\omega_k\tau$.
This diagonal evolution, along with quantum circuits for fermionic Fourier and Bogoliubov transforms, provides a quantum circuit for the exact evolution of the Ising model~[\cref{fig:qcircs}(a)].

We use the tailored circuit in \cref{fig:qcircs}(b) for the exact evolution of the 2-site Ising model $H = JX_0X_1 + h(Z_0 + Z_1)$ for our hardware implementations.
This circuit follows from the diagonalization of $H = U D U^\dag$, where $U$ is the diagonalizing unitary that we implement using the gates inside the dashed box in \cref{fig:qcircs}(b) and
\begin{equation}
    D = \omega_0 Z_0 + \omega_1 Z_1\quad
    \omega_{0/1}:= \frac12(\lambda_0\pm \lambda_1)
    = \frac12\left(\sqrt{J^2+(2h)^2}\pm J\right)
\end{equation}
is the diagonalized Hamiltonian with $\lambda_{0/1}$ the positive eigenvalues of $H$.

\begin{figure}
    \centering
    \includegraphics[width=\linewidth]{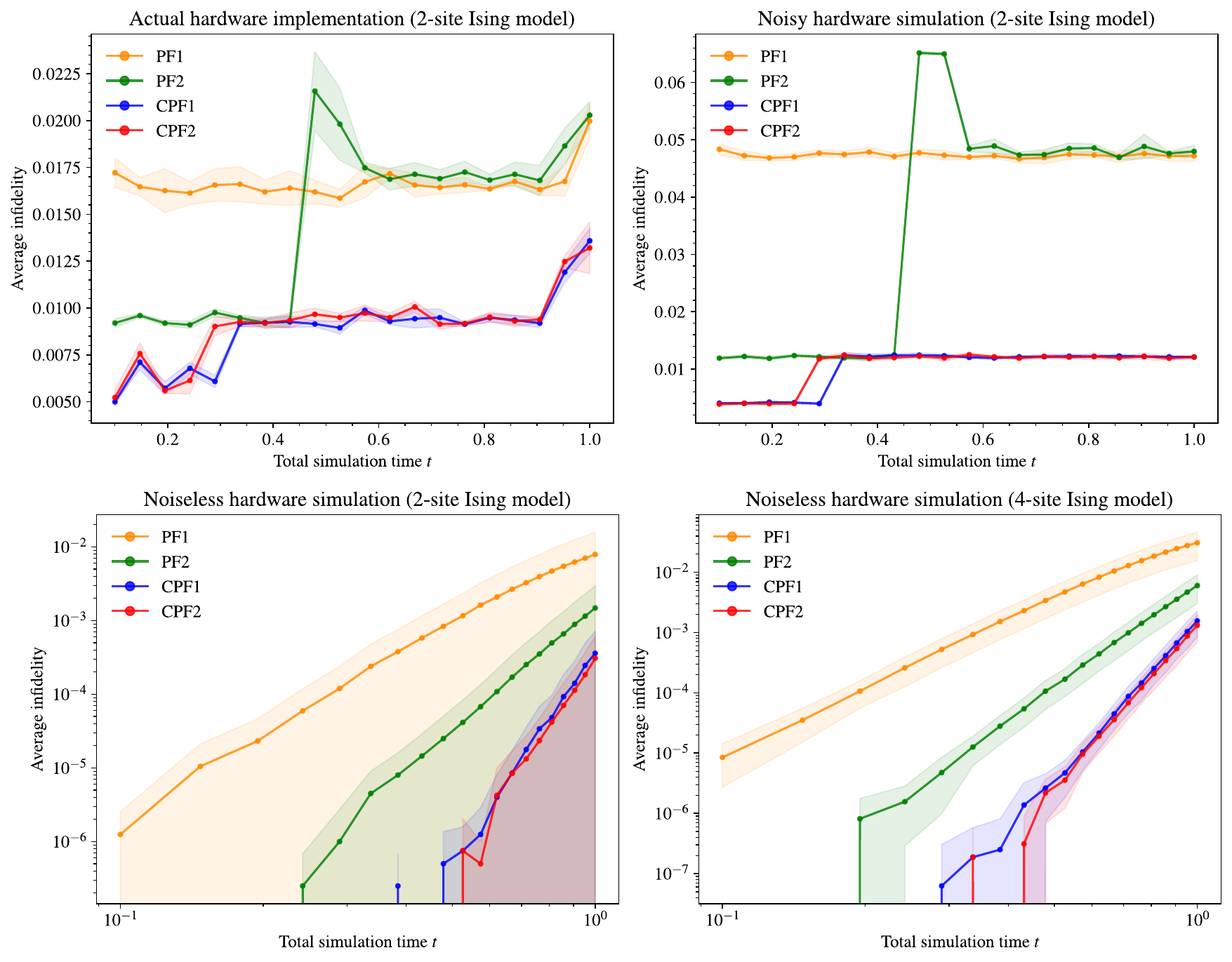}
    \caption{
    Error of CPFs with symplectic correctors vs standard PFs for simulating the Ising model with a weak coupling on actual quantum hardware (127-qubit \texttt{ibm\_quebec} QPU), as well as noisy (127-qubit \texttt{FakeQuebec}) and noiseless (\texttt{Qiskit Aer}) hardware simulators.
    Average infidelity is evaluated at 20 evenly spaced points in $0.1\leq t\leq1$; shaded area around each line represents one standard deviation above and below the average.
    The perturbation parameter $\alpha=0.1$ is used in all experiments.
    $r=10$ time steps and~$s=10^5$ shots are used for the hardware implementation and  noisy simulation. For noiseless simulation, we use $r=1$ and $s=10^6$.
    }
    \label{fig:hardware}
\end{figure}
\begin{figure}
    \centering
    \includegraphics[width=.9\linewidth]{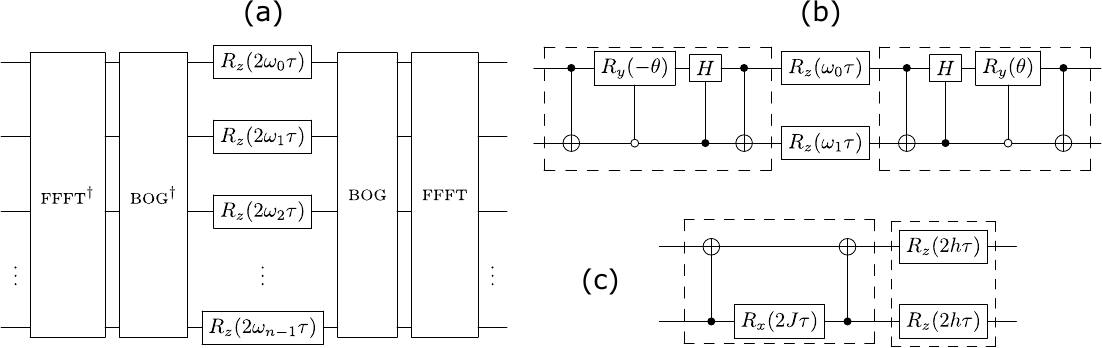}
    \caption{Quantum circuits for the Ising model.
    (a)~Circuit for exact evolution of $n$-site model by fast fermionic Fourier transform~(FFFT) and Bogoliubov transform~(BOG).
    (b)~Tailored circuit for the $2$-site model. Gates inside the (left) right dashed box implement the (inverse) unitary that diagonalizes the Hamiltonian. (c)~Circuit for PF1 with gates in the left (right) box implementing the evolution by the interaction (external field) term. 
    }
    \label{fig:qcircs}
\end{figure}

%==================================================
\subsection{Results of hardware implementations}
\label{subsec:hardware_results}
%==================================================
For hardware implementation, we use the average infidelity defined as
\begin{equation}
    \label{eq:infid}
    \text{Average infidelity} =
    \mathbb{E}_{\{\ket{x}\}}
    \left[1-\left|\bra{x} U_\text{exact}^\dagger U_\text{approx}\ket{x}\right|^2\right]
\end{equation}
to quantify the approximation error.
Here $\mathbb{E}_{\{\ket{x}\}}$ denotes the arithmetic average over all computational basis states of the system under evolution.
The exact unitary in our implementations is the exact evolution of the Ising model, and the approximate unitary is a corrected or standard product formula.
For hardware experiments, we compare average infidelities of low-order CPFs with symplectic correctors against the standard PFs of the same order. 
Specifically, we use CPF1 with the symplectic corrector in the second row of PF1 correctors in \cref{tab:correctors} and CPF2 with the symplectic corrector in the first row of PF2 correctors in \cref{tab:correctors}.
For convenience, we express these correctors as
\begin{equation}
\label{eq:correctors_hardware}
    C^{\text{PF1}}_\text{symp} =
    -\frac{\i}2\alpha\tau B
    -\frac1{12}\alpha\tau^2[A,B],
    \quad
    C^{\text{PF2}}_\text{symp} = \frac1{24}\alpha\tau^2[A,B]
\end{equation}
with $\lambda=-\i\tau$ for Hamiltonian evolution.
We use the compilations in \cref{tab:compilations} for hardware implementations of correctors.

In all experiments, we use the transverse-field Ising model and fix~$h=1$ and $J=\alpha=0.1$, with $\alpha$ the perturbation parameter. This weakly coupled model can be expressed as $H=H_z + \alpha H_{xx}$ with $H_{xx}(H_z)$ the coupling (external-field) term.
Due to hardware limitations, we restrict the system size (2-site and 4-site Ising model), the number of timesteps ($r=10$), and the evolution time ($0.1\leq t\leq1$) for the actual and noisy hardware implementations.
Given the hardware limitations, we use $10^5$ samples for computing the infidelity for each computational basis.
The results of our hardware implementations are presented in \cref{fig:hardware}, showing that errors of~CPFs are smaller than those of PFs.

The quantum circuits in~\cref{fig:qcircs} are transpiled to the native gates of \texttt{ibm\_quebec} QPU for hardware implementations.
\cref{tab:circuit_depth} provides the depth of the transpiled circuits at different optimization levels of Qiskit's compiler.
Due to the limited number of steps~$r$ used for hardware implementations, the depth of transpiled circuits for CPFs is significantly higher than those of PFs, especially at the lowest optimization level. The additional circuit depth due to symplectic correctors becomes negligible as~$r$ increases, as it does not scale with~$r$.
We note that no error mitigation techniques were used in our hardware implementations for CPFs and PFs. Incorporating these techniques could improve the~results.
\begin{table}
\centering
\begin{tabular}{cc|cccccc} 
 \toprule
  \multirow{2}{*}{Size} &\multirow{2}{*}{Time steps} &Optimization &\multirow{2}{*}{PF1~(Infid)} &\multirow{2}{*}{PF2~(Infid)} &\multirow{2}{*}{CPF1~(Infid)} &\multirow{2}{*}{CPF2~(Infid)} &Exact\\ 
   & &level for circuits& &&& & evolution\\
 \midrule
 \multirow{3}{*}{2}&\multirow{3}{*}{10} &Level 1 &91~(138) &91~(138) &154~(200) &145~(192) &49\\ 
 &&Level 2 &12~(12) &12~(12) &12~(12) &12~(12) &12\\
 &&Level 3 &12~(12) &12~(12) &12~(12) &12~(12)&12\\
 \midrule
 \multirow{3}{*}{4}&\multirow{3}{*}{1} &Level 1 &34 (263) &34 (263) &296 (526) &232(461) & 231\\
 &&Level 2 &48 (130) &46 (129) &389 (473) &309 (393) &77\\
 &&Level 3 &48 (111) &46 (110) &389 (454) &309 (374) &57\\
 \bottomrule
\end{tabular}
\caption{\label{tab:circuit_depth}
Depth of the circuits transpiled to the native gates of \texttt{ibm\_quebec} QPU for the exact evolution, the approximate evolutions by corrected and standard product formulas, and the average infidelity~(Infid) circuits for the 2- and 4-site Ising model used for hardware experiments in~\cref{fig:hardware}. Duration of each timestep is~$\tau = 0.1$, and the total evolution time for~$r$ steps is $r \times \tau$.
Transpiled circuits are optimized using Qiskit's compiler at different optimization levels: Level~1 provides a light optimization, Level~2 provides a medium optimization, and Level~3 provides a heavy optimization~\cite{transpile}.
Level-3 optimization is used for hardware experiments in~\cref{fig:hardware}.
}
\end{table}

%==================================================
\section{Discussion}
\label{sec:discussion}
%==================================================
Product formulas have been the original proposal for simulating Hamiltonian evolution on a quantum computer. Despite the development of several other promising quantum algorithms for this task in recent years, the conventional approach based on product formulas remains competitive for practical applications.
In this work, we developed high-order corrected product formulas~(CPFs) based on three types of correctors and established theoretical results showing that CPFs improve the error bounds of standard product formulas by orders of magnitude, leading to a substantial reduction in the gate cost for Hamiltonian simulation.
The correctors we developed are based on a linear combination of nested commutators, and we presented a procedure for compiling them using Hamiltonian terms. Our approach for compiling nested commutators has applications beyond CPFs; it can also be used to efficiently synthesize complicated unitaries on a quantum simulator with a limited set of native gates~\cite{CW13,CCH22}.

To verify the established error bounds and demonstrate the performance of CPFs, we performed numerical simulations for various (non-)perturbed lattice Hamiltonians.
Our numerical results show that the theoretical performance of CPFs matches or exceeds the empirical performance of standard product formulas, which is often much better than known theoretical bounds.
We also complemented our theoretical and numerical results with the implementations of CPFs on actual quantum hardware, as well as on both noisy and noiseless quantum simulators, demonstrating the improvements that CPFs can provide in using current and near-term quantum computers.

We applied correctors to construct CPFs for simulating perturbed ($\alpha\ll1$) and non-perturbed $(\alpha=1)$ systems with a Hamiltonian of the form $H=A+\alpha B$, where $A$ and $B$ have comparable norms.
For non-perturbed systems, we assume both partitions, $A$ and $B$, can be exactly simulated, while for perturbed systems, we assume only the main partition, $A$, is exactly simulatable.
The CPF of order~$2k$ we constructed for non-perturbed systems achieves an error bound of $\order{t^{2k+3}}$ providing two orders of magnitude improvements for the error bound of the standard product formula with the same order.
CPFs, however, are more advantageous for perturbed systems.
In particular, we established a CPF of order~$2k$ for such systems that achieves the error bound $\order{\alpha^2t^{2k+1}}$, which is a factor of~$\alpha$ better than that for the standard product formula of the same order.
Furthermore, we established several customized low-order CPFs summarized in \cref{tab:correctors} that provide orders of magnitude reduction in the error bound of low-order product formulas.
Similar to low-order standard product formulas, the low-order CPFs are preferred in practical applications as high-order product formulas have a prefactor that grows rapidly with the order parameter~$2k$.

The Hamiltonian form we considered in developing correctors is a common characteristic of lattice Hamiltonians.
These Hamiltonians typically can be divided into two exactly simulatable parts because either they contain pairwise commuting terms or they can be efficiently diagonalized, as seen in the example Hamiltonians in~\cref{sec:numerics}. The assumptions taken for perturbed systems apply for generic Hamiltonians of the form $H=T+V$ with~$T$ the kinetic part that is exactly simulatable and~$V$ a weak potential part with a small norm.
A prime example for this case is the electronic-structure Hamiltonian represented in the first-quantized plane wave basis in the regime where the number of electrons~$\mu$ is much smaller than the number of plane wave orbitals~$N$~($\mu \ll N$): the norm of the kinetic part in this regime is much smaller than the norm of the sum of the potential parts~\cite{BBM+19}.

While we demonstrated the advantage of CPFs in simulating Hamiltonians with exactly simulatable partitions, we remark that CPFs are not limited to such cases. CPFs could also be advantageous even in cases where none of the partitions are exactly simulatable.
As in the divide and conquer approach for Hamiltonian simulation~\cite{HP18}, CPFs can be used to simulate the evolution generated by such Hamiltonians in terms of exponentials of the Hamiltonian partitions, each of which can be subsequently approximated by standard product formulas.
We expect this approach to be more advantageous for cases where one partition has a significantly smaller norm than the other.
A prime example for such cases is a Pauli representation of the electronic-structure Hamiltonian in computational chemistry, such as the minimal basis set.
The spectral norm distribution of the Hamiltonian terms (i.e., the distribution of the magnitude of the Pauli coefficients) in this case is sharply peaked~\cite{XS24,HP18,PHC+14}. By a hard cutoff on the spectral norm, for instance, this Hamiltonian can be divided into two partitions, one with few terms and a large norm and the other with many terms and a small norm.

The high-order CPFs we developed are built from a CPF2 with a symplectic corrector for perturbed systems and a symmetric corrector for non-perturbed systems. We have also constructed a fourth-order CPF with a symplectic corrector. A topic for future work is to develop higher-order CPFs with a symplectic corrector. CPFs with symplectic correctors are preferred as these correctors cancel out in the intermediate simulation steps, resulting in a small additive cost to the total simulation cost.
As for the standard product formulas, we used the spectral-norm error as the measure of error for CPFs. However, recent work~\cite{MCP+24} suggests that the eigenvalue error is a more appropriate error measure for product formulas. Comparing CPFs with standard product formulas using eigenvalue error is another avenue for future research.
One especially appealing feature of product formulas is that they often admit substantial improvements over other Hamiltonian simulation techniques when the error is evaluated in specific cases, such as for specific input states~\cite{SR21,AFL21,ZZS+22,ZZC25,BGH+23,BFH+24} or observables~\cite{HHZ19,BF25,CST+21}, rather than in worst-case settings.
CPFs could similarly benefit from improved performance under these more practical error measures, making this an especially promising direction for further investigation.

\section*{Data Availability}
The code and data used to generate the results reported in this manuscript are publicly available at Ref.~\cite{codes}.

%==================================================
\section{Acknowledgments}
%==================================================
We thank the anonymous referees for their detailed feedback and helpful suggestions.
MB and AAG acknowledge the generous support and funding of this project by the Defense Advanced Research Projects Agency (DARPA) under Contract No. HR0011-23-3-0021.
Any opinions, findings, and conclusions or recommendations expressed in this material are those of the author(s) and do not necessarily reflect the views of the Defense Advanced Research Projects Agency.
MB and AAG also acknowledge partial funding of this project by the National Sciences and Engineering Research Council of Canada (NSERC) Alliance Consortia Quantum Grants \#ALLRP587590-23.
AAG also acknowledges support from the Canada 150 Research Chairs program and NSERC-IRC. AAG also acknowledges the generous support of Anders G. Frøseth.
AA gratefully acknowledges King Abdullah University of Science and Technology (KAUST) for the KAUST Ibn Rushd Postdoctoral Fellowship.
DWB worked on this project under a sponsored research agreement with Google Quantum AI. DWB is also supported by Australian Research Council Discovery Projects DP210101367 and DP220101602.

M.~B. and L.~M.~C. contributed equally to the hardware implementations.

%==================================================
\appendix
%==================================================
\section{Proofs}
\label{apx:proofs}
%==================================================
Here we provide a proof for the formula in~\cref{eq:PF1W} used for the PF1 corrector and the symplectic corrector given in \cref{eq:CPF2symp} for PF2.
We begin by proving~\cref{eq:PF1W} for the corrector $C= \lambda B/2 + c_2 \lambda^2 [A,B]$ used for PF1. By \cref{eq:sympC} we have $\e^C S_1(\lambda) \e^{-C} = \e^{K'_1}$ with $K'_1 = K_1 + [C,K_1] + \frac12[C,C,K_1] + \order{|\lambda|^4}$ where
\begin{equation}
    K_1=\lambda (A+B)+\frac12\lambda^2[A,B]
    +\frac1{12}\lambda^3[A-B,A,B]
    + \order{|\lambda|^4}
\end{equation}
is the kernel of PF1.
We expand the first commutator in $K'_1$ as
\begin{align}
    [C,K_1] &= [(\lambda/2) B + c_2 \lambda^2 [A,B], \lambda (A+B)+(\lambda^2/2)[A,B] + \order{|\lambda|^3}]\\
    &=-\frac12\lambda^2[A,B]
    +(c_2-\frac14)\lambda^3 [B,B,A]
    -c_2\lambda^3[A,A,B] + \order{|\lambda|^4},
\end{align}
and the second commutator as
\begin{align}
    \frac12[C,C,K_1]
    = \frac12[\lambda B/2 + \order{|\lambda|^2},\,
    \lambda B/2 + \order{|\lambda|^2},\,
    \lambda(A+B)+\order{|\lambda|^2}]
    = \frac18\lambda^3[B,B,A]+\order{|\lambda|^4}.
\end{align}
Altogether, we have
\begin{equation}
    K'_1 = \lambda(A+B)+(\tfrac1{12}-c_2)\lambda^3[A,A,B]+
    (c_2-\tfrac1{24})\lambda^3
    [B,B,A] + \order{|\lambda|^4}
\end{equation}
for the modified kernel.

\textbf{PF2 corrector.}
We now prove the expression of the symplectic corrector given in \cref{eq:CPF2symp} for PF2.
Using Proposition~\ref{prop:largeterms} with $s=1/2$ and noting that Bernoulli polynomials $B_j(x)=0$ at $x=1/2$ for odd $j$,
we have
\begin{equation}
    \e^{\lambda A/2}\e^{\lambda\alpha B}\e^{-\lambda A/2} = \e^{K_2} 
    \quad
    \text{with}
    \quad
    K_2 \equiv_{(\geq2)} \lambda (A + B) + \sum_{j=1}^{\infty} \frac{B_{2j}(1/2)}{(2j)!}\lambda^{2j} \ad^{2j}_A(B),
\end{equation}
where $\equiv_{(\geq2)}$ denotes equality modulo terms with degree $\geq 2$ in~$B$, and $B_{j}(x)$ are Bernoulli polynomials with few nonzero values
\begin{equation}
\label{eq:bernoulli}
    B_0 = 1,\;
 B_2 = \frac{-1}{12},\;
 B_4 = \frac{7}{240},\;
 B_6 = \frac{-31}{1344},\;
 B_8 = \frac{127}{3840},\;
 B_{10} = \frac{-2555}{33\,792}
\end{equation}
at~$x=\tfrac12$.
We note that the full expression of~$K_2$ (without modulo) is an infinite series that is convergent for
\begin{equation}
\label{eq:constraint}
    |\lambda|(\norm{A}+\alpha\norm{B})\leq \frac12 \log2
\end{equation}
by \cite[Theorem 2.2]{BC04}.
A symplectic corrector $C$ modifies $K_2$ according to \cref{eq:sympC} as $K'_2 = K_2 + [C,K_2] + \cdots$.
We want $K'_2 \equiv_{(\geq2)} \lambda (A+B)$.
Observe that
\begin{equation}
    C = \sum_{j=1}^{\infty} \frac{B_{2j}(1/2)}{(2j)!} \lambda^{2j}\,\ad^{2j-1}_A(B)
    \quad \text{yields}\quad
    [C,K_2] \equiv_{(\geq2)} - \sum_{j=1}^{\infty} \frac{B_{2j}(1/2)}{(2j)!} \lambda^{2j}\,\ad^{2j}_A(B).
\end{equation}
By this corrector, we have $ K'_2 \equiv_{(\geq2)} A+B$.
The corrector $C(k)$ in \cref{eq:CPF2symp} is obtained by truncating the series at~$j\leq k$.
We note that both the full corrector~$C$ and its truncated version $C(k)$ are convergent under the condition in \cref{eq:constraint}, because the series for $K_2$ is itself convergent.

%==================================================
\section{Extended numerical simulations}
\label{app:extended_simulations}
%==================================================
In this appendix, we present numerical simulations for the two-dimensional~(2D) versions of the perturbed and non-perturbed systems considered in the main text, as well as for higher-order product formulas. We also provide additional simulations for 1D systems over a broader parameter range to further demonstrate the generality of the improvements achieved by CPFs.
Specifically, we consider larger 1D systems~($n=10$), longer simulation times~($1\leq t\leq100$), and a larger number of simulation steps~($r=1000$).

The results of fixed-time-step simulations for 2D systems are shown in~\cref{fig:2d_fixed_tau}.
\cref{fig:2dsimulations} shows the results for 2D systems in simulations with a fixed number of time steps.
The results for non-perturbed and perturbed systems of size~$n=10$ are shown in \cref{fig:nonpert_size10} and \cref{fig:pert_size10}, respectively.
Additional simulation details are provided in the captions of these figures.
Simulations for systems with $n=12$ show similar behavior and are omitted for brevity.
Overall, these results confirm that CPFs consistently outperform the standard product formulas across a broad range of systems and simulation parameters.

\begin{figure}[!b]
    \centering
    \includegraphics[width=\linewidth]{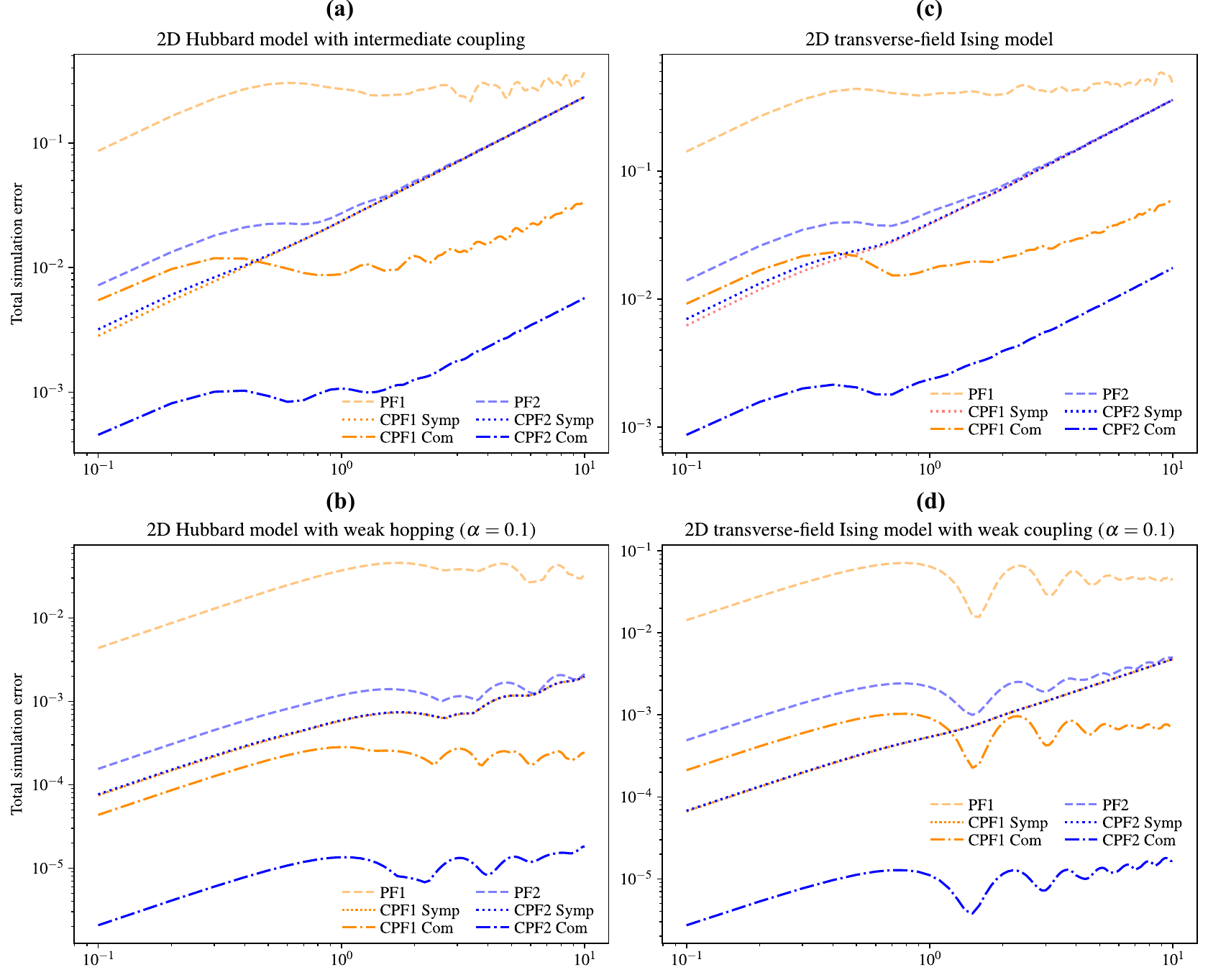}
    \caption{Empirical errors of PFs~(dashed lines) and CPFs (dotted and dash-dotted lines) for non-perturbed~[top:~(a),~(c)] and perturbed~[bottom:~(b),~(d)] systems in simulations with a fixed time step.
    Each system is on a 2D lattice of size~$3\times3$ with open boundaries.
    Correctors used are the same as those in \cref{fig:fixed_tau_simulations}.
    Errors are evaluated at fixed $\tau=0.1$ for the number of steps $1\leq r\leq100$, so $t=r\tau$ ranges over $0.1<t<10$.
    Observe that CPFs are more advantageous for perturbed systems~($\alpha<1$).
    }
    \label{fig:2d_fixed_tau}
\end{figure}

\begin{figure}
    \centering
    \includegraphics[width=\linewidth]{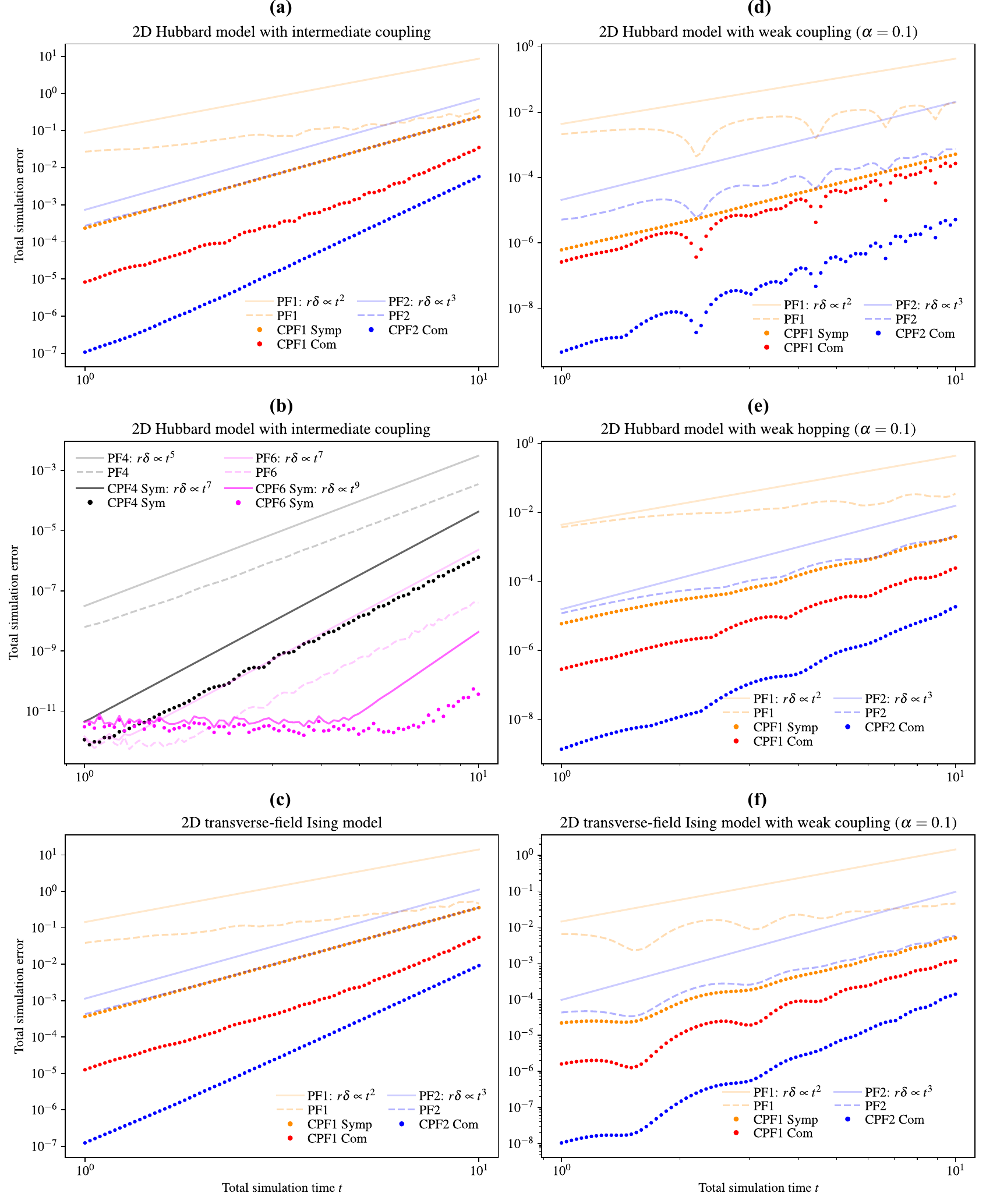}
    \caption{Empirical errors of PFs (dashed lines) and CPFs (dotted lines) for 2D non-perturbed~[left:~(a)--(c)] and perturbed~[right:~(d)--(f)] systems.
    Each system is on a 2D lattice of size~$3\times3$ with open boundaries.
    Correctors used are the same as those in \cref{fig:fixed_tau_simulations}.
    Errors are numerically evaluated at time $1\leq t\leq 10$ using $r=100$ simulation steps; $\delta$ is the one-step empirical error, and the plots of~$r\delta$ (solid lines) show the theoretical scaling of the total error obtained from the triangle inequality.
    } 
    \label{fig:2dsimulations}
\end{figure}

\begin{figure}
    \centering
    \includegraphics[width=\linewidth]{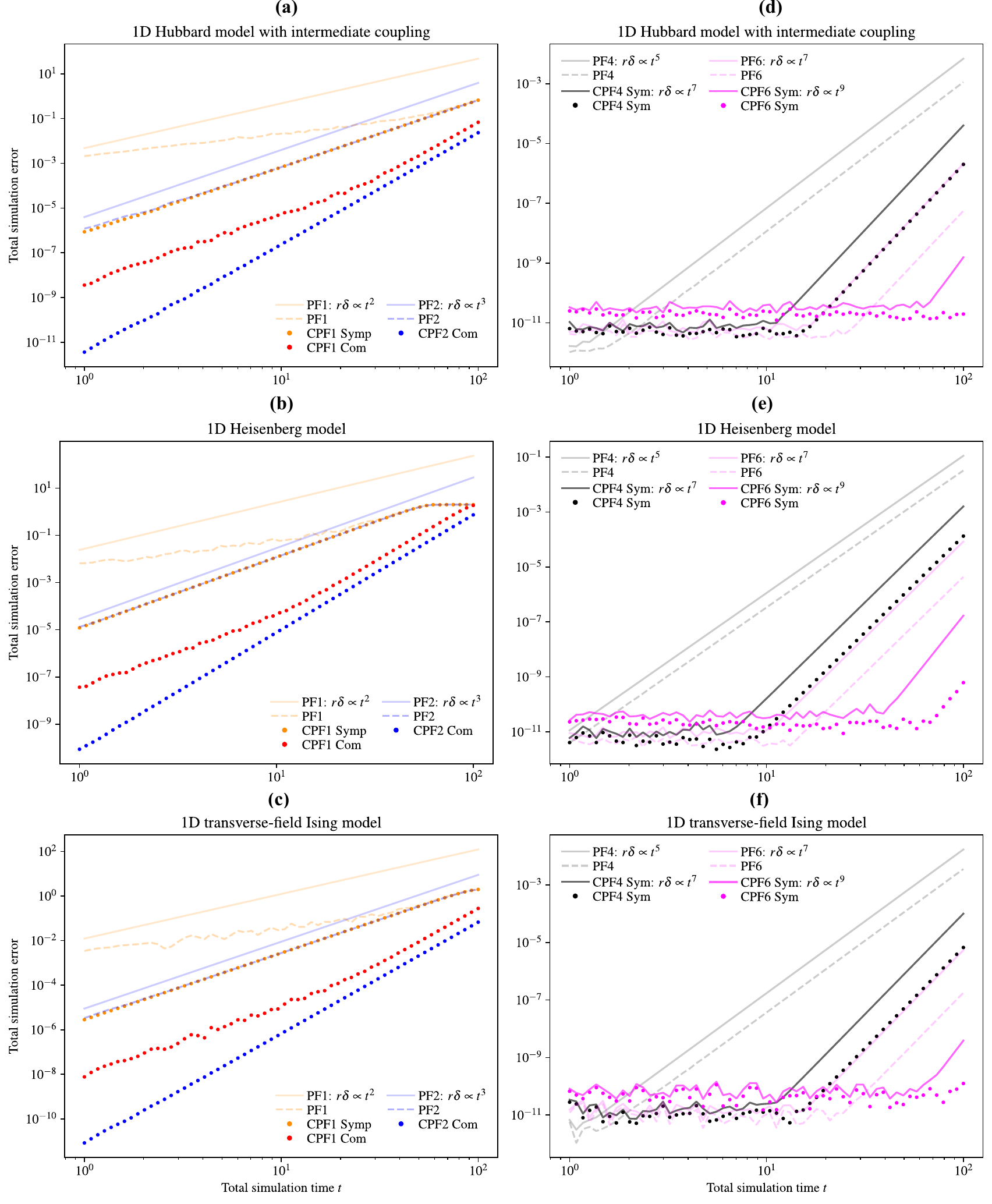}
    \caption{Empirical errors of PFs (dashed lines) and CPFs (dotted lines) for simulating non-perturbed systems.
    Each system is on a 1D lattice of size $n=10$ with periodic boundaries.
    Correctors used are the same as those in \cref{fig:fixed_tau_simulations}.
    Errors are numerically evaluated for $1\leq t\leq 100$ using~$r=1000$ simulation steps; $\delta$ is the empirical error in one step, and the plots of~$r\delta$ (solid lines) show the theoretical scaling of the total error.
    }
    \label{fig:nonpert_size10}
\end{figure}

\begin{figure}
    \centering
    \includegraphics[width=\linewidth]{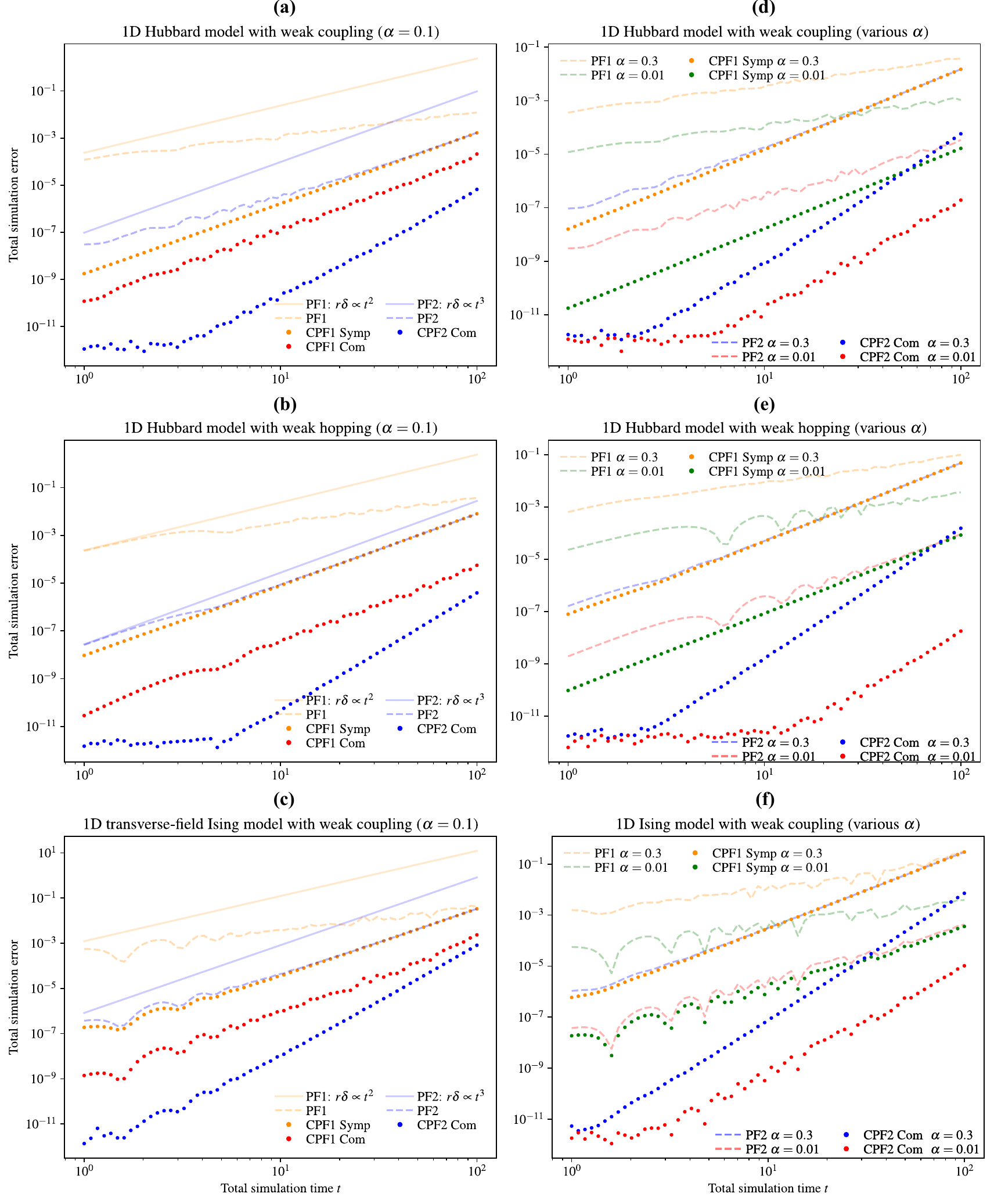}
    \caption{Empirical errors of PFs (dashed lines) and CPFs (dotted lines) for perturbed systems with fixed~(left:~a,~b,~c) and various perturbation parameter~$\alpha$~(right:~d,~e,~f).
    Each system is on a 1D lattice of size~$n=10$ with periodic boundaries.
    Correctors used are the same as those in \cref{fig:fixed_tau_simulations}.
    Errors are numerically evaluated for $1\leq t\leq 100$ using~$r=1000$ simulation steps; $\delta$ is the one-step empirical error, and the plots of~$r\delta$ (solid lines) show the theoretical scaling of the total error.} 
    \label{fig:pert_size10}
\end{figure}

%==================================================
\bibliographystyle{apsrev4-2}
\bibliography{ref.bib}
%==================================================
%==================================================
%==================================================
\end{document}